\documentclass[runningheads,envcountsame, 10pt]{llncs}

\usepackage[T1]{fontenc}
\usepackage[english]{babel}

\usepackage{graphicx}

\usepackage{amsfonts, amsthm, amsmath,amssymb}

\usepackage{hyperref}
\usepackage[capitalize,nameinlink,noabbrev]{cleveref}

\usepackage{mathtools}

\usepackage{mathrsfs}

\usepackage[all]{xy}

\usepackage{xfrac}
\usepackage{faktor}

\usepackage[nice]{nicefrac}

\usepackage[shortlabels]{enumitem}

\usepackage{stmaryrd}

\usepackage{xspace}

\usepackage{tikz}
\usetikzlibrary{patterns, arrows, decorations.pathmorphing, decorations.pathreplacing, decorations.markings, decorations.text , 3d, calc, cd, positioning}

\usepackage{dcolumn}

\usepackage{stackrel}

\usepackage{colonequals}

\usepackage[all]{xy}

\usepackage{lineno}

\usepackage{csquotes}

\usepackage{bussproofs}

\usepackage{chngcntr}

\usepackage{extarrows}

\usepackage{multicol}

\usepackage{float}

\usepackage{lmodern}

\usepackage{comment}

\usepackage{authblk}

\usepackage{bussproofs}

\usepackage{tabularx}

\usepackage{breakcites}

\newtheoremstyle{plain}{15pt}{15pt}{\itshape}{}{\bfseries}{.}{.5em}{}
\newtheoremstyle{definition}{15pt}{20pt}{}{}{\bfseries}{.}{.5em}{}

\theoremstyle{plain}

\theoremstyle{definition}

\theoremstyle{remark}

\DeclarePairedDelimiter\set{\{}{\}}
\DeclarePairedDelimiterX\setvbar[2]{\{}{\}}{#1 \nonscript\;\delimsize \vert \nonscript\; #2}
\DeclarePairedDelimiterX\setcolon[2]{\{}{\}}{#1 : #2}

\DeclarePairedDelimiter{\Card}{\lvert}{\rvert} 

\newcommand{\ov}[1]{\overline{\vphantom{1}#1}}

\DeclareMathOperator{\ima}{im}

\DeclareMathOperator{\Id}{id}
\DeclareMathOperator{\id}{id}

\let\emptyset\varnothing

\renewcommand{\leq}{\leqslant}
\renewcommand{\geq}{\geqslant}

\newcommand{\N}{\ensuremath{\mathbb{N}}\xspace}

\let\isom\cong

\newcommand*{\defeq}{\mathrel{\vcenter{\baselineskip0.5ex \lineskiplimit0pt
                    \hbox{\scriptsize.}\hbox{\scriptsize.}}}%
                    =}

\let\lnottemp\lnot
\renewcommand{\lnot}{\lnottemp \hspace*{0.1em}}

\let\oldexists\exists
\let\exists\relax

\newcommand{\exists}{\hspace*{0em}\oldexists\hspace*{0.07em}}
\let\oldforall\forall
\let\forall\relax 

\newcommand{\forall}{\hspace*{0em}\oldforall\hspace*{0.07em}}

\tikzcdset{scale cd/.style={every label/.append style={scale=#1},
    cells={nodes={scale=#1}}}}

\newcommand{\cat}[1]{\mathsf{#1}}
\newcommand{\catset}{\cat{Set}}
\newcommand{\EM}[1]{{\normalfont\textbf{EM}(#1)}} 
\newcommand{\catalg}[1]{{\normalfont\textbf{Alg}(#1)}} 
\newcommand{\catmon}[1]{{\normalfont\textbf{Mon}(\cat{#1})}}
\newcommand{\inlsym}{{\mathsf{inl}}}
\newcommand{\inrsym}{{\mathsf{inr}}}

\newcommand{\inl}{\ensuremath\operatorname{\inlsym}\xspace}
\newcommand{\inr}{\ensuremath\operatorname{\inrsym}\xspace}

\newcommand{\singleset}{\textbf{1}}
\newcommand{\singl}{\ensuremath{\set{\ast}}}

\newcommand{\inv}{^{\raisebox{.2ex}{$\scriptscriptstyle-1$}}}

\newcommand\Item[1][]{%
  \ifx\relax#1\relax  \item \else \item[#1] \fi
  \abovedisplayskip=0pt\abovedisplayshortskip=0pt~\vspace*{-\baselineskip}}

\newcommand{\newa}{\mathsf{a}}
\newcommand{\newb}{\mathsf{b}}
\newcommand{\op}{\mathsf{op}}
\newcommand{\variables}{\mathsf{Var}}
\newcommand{\EMs}[1]{{\normalfont\textbf{EM}_s(#1)}} 

\newcommand{\terms}[2]{\mathcal{T}_{#1}(#2)}
\newcommand{\freealgebra}[3]{\nicefrac{\terms{#1}{#2}}{#3}}
\newcommand{\freemonad}[1]{T_{\scriptscriptstyle #1}}
\newcommand{\freemonadunit}{\eta_{\scriptscriptstyle \Sigma,E}}
\newcommand{\freemonadmult}{\mu_{\scriptscriptstyle \Sigma,E}}
\newcommand{\cp}{\mathsf{dp}}
\newcommand{\freeinter}[2]{\Brackets{#2}^{#1}}
\newcommand{\freeinterX}[1]{\freeinter{X}{#1}}
\newcommand{\freeinterY}[1]{\freeinter{Y}{#1}}

\newcommand{\freeinterTX}[1]{\freeinter{TX}{#1}}

\DeclareTextFontCommand{\myemph}{\bfseries\em}

\makeatletter
\DeclareFontEncoding{LS2}{}{\@noaccents}
\makeatother
\DeclareFontSubstitution{LS2}{stix}{m}{n}

\DeclareSymbolFont{largesymbolsstix}{LS2}{stixex}{m}{n}

\DeclareMathDelimiter{\lParen}{\mathopen}{largesymbolsstix}{"DE}{largesymbolsstix}{"02}
\DeclareMathDelimiter{\rParen}{\mathclose}{largesymbolsstix}{"DF}{largesymbolsstix}{"03}
\DeclareMathDelimiter{\lBrack}{\mathopen}{largesymbolsstix}{"E0}{largesymbolsstix}{"06}
\DeclareMathDelimiter{\rBrack}{\mathclose}{largesymbolsstix}{"E1}{largesymbolsstix}{"07}
\DeclareMathDelimiter{\lBrace}{\mathopen}{largesymbolsstix}{"E8}{largesymbolsstix}{"0E}
\DeclareMathDelimiter{\rBrace}{\mathclose}{largesymbolsstix}{"E9}{largesymbolsstix}{"0F}
\DeclareMathDelimiter{\lbrbrak}{\mathopen} {largesymbolsstix}{"EE}{largesymbolsstix}{"14}
\DeclareMathDelimiter{\rbrbrak}{\mathclose}{largesymbolsstix}{"EF}{largesymbolsstix}{"15}

\DeclareFontFamily{OMX}{MnSymbolE}{}
\DeclareFontShape{OMX}{MnSymbolE}{m}{n}{
    <-6>  MnSymbolE5
   <6-7>  MnSymbolE6
   <7-8>  MnSymbolE7
   <8-9>  MnSymbolE8
   <9-10> MnSymbolE9
  <10-12> MnSymbolE10
  <12->   MnSymbolE12}{}
\DeclareFontShape{OMX}{MnSymbolE}{b}{n}{
    <-6>  MnSymbolE-Bold5
   <6-7>  MnSymbolE-Bold6
   <7-8>  MnSymbolE-Bold7
   <8-9>  MnSymbolE-Bold8
   <9-10> MnSymbolE-Bold9
  <10-12> MnSymbolE-Bold10
  <12->   MnSymbolE-Bold12}{}
 
\DeclareSymbolFont{largesymbolsmnsymbol}  {OMX}{MnSymbolE}{m}{n}

\DeclareMathDelimiter{\langlebar}{\mathopen}{largesymbolsmnsymbol}{'152}{largesymbolsmnsymbol}{'152}
\DeclareMathDelimiter{\ranglebar}{\mathclose}{largesymbolsmnsymbol}{'157}{largesymbolsmnsymbol}{'157}
\DeclarePairedDelimiter{\Brackets}{\lBrack}{\rBrack}
\DeclarePairedDelimiter{\Parens}{\langlebar}{\ranglebar}

\newcommand{\supp}[1]{\mathsf{supp(#1)}}

\newcommand{\multiset}{{\mathcal{M}}}
\newcommand{\powerset}{{\mathcal{P}}}
\newcommand{\distribution}{{\mathcal{D}}}

\newcommand{\state}{{\mathsf{State}}}

\newcommand{\semifree}{^{\mathrm{s}}}
\newcommand{\nsemifree}[1]{^{\mathrm{s}^{#1}}}
\newcommand{\rem}[1]{}  

\newcommand{\cisom}{\isom_{\textsf{conc}}}

\colorlet{linkcolor}{red!60!black}
\hypersetup{
    colorlinks=true,
    linkcolor=linkcolor,
    citecolor=linkcolor,
    filecolor=linkcolor,      
    urlcolor=linkcolor,
    pdftitle={Algebraic Presentation of Semifree Monads},
    pdfpagemode=FullScreen,
}

\allowdisplaybreaks

\begin{document}
\title{\texorpdfstring{Algebraic Presentation of Semifree Monads}{Algebraic Presentation of Semifree Monads}}


\author{Alo\"is Rosset \inst{1} \and
Helle Hvid Hansen \inst{2} \and 
J\"org Endrullis \inst{1}
}

\authorrunning{A.~Rosset, H.H.~Hansen, J.~Endrullis} 
\titlerunning{Algebraic Presentation of Semifree Monads}
\institute{Vrije Universiteit Amsterdam, Amsterdam, Netherlands \\
\email{\{alois.rosset, j.endrullis\}@vu.nl}
\and 
University of Groningen, Groningen, Netherlands \\
\email{h.h.hansen@rug.nl}
}

\maketitle

\begin{abstract}
\rem{
We study algebraic presentation of monads.
Given a monad $M$ there is a suitable monad structure on the
functor id + M, which is called the semifree monad $M\semifree$.
This $M\semifree$ is of importance for two reasons.
First, there is an isomorphism of categories between $M\semifree$-algebras and $M$-semialgebras.
Second, distributive laws of $M$ over another monad $T$ are in bijection with a subclass of all the distributive laws of $M\semifree$ over $T$.
Our main contribution is a uniform algebraic presentation of those semifree monads $M\semifree$.
}

Monads and their composition via distributive laws have many applications in program semantics and functional programming. For many interesting monads, distributive laws fail to exist, and this has motivated investigations into weaker notions. In this line of research, Petri\c{s}an and Sarkis recently introduced a construction called the semifree monad in order to study semialgebras for a monad and weak distributive laws.
In this paper, we prove that an algebraic presentation of the semifree monad $M\semifree$ on a monad $M$ can be obtained uniformly from an algebraic presentation of $M$. This result was conjectured by Petri\c{s}an and Sarkis. 
We also show that semifree monads are ideal monads, that the semifree construction is not a monad transformer, and that the semifree construction is a comonad on the category of monads.

\end{abstract}

\section{Introduction}

Monads \cite{MacLane_1971,Manes_1976} are widely used in the coalgebraic semantics of programs with nondeterminism, probabilistic branching and other features,
see e.g.~\cite{Bonchi_Sokolova_Vignudelli_2019,Heerdt_Sammartino_Silva_2020,Jacobs_Silva_Sokolova_2014,Katsumata_Rivas_Uustalu_2020,Milius_Pattinson_Schroeder_2015}.
In functional programming, monads are used
for structuring programs with computational effects, see e.g. \cite{Moggi_1991,Plotkin_Power_2001,Wadler_1992}. 
In order to reason about programs that combine several effects, compositions of monads have been studied such as monad transformers~\cite{Jaskelioff_09,Liang_Hudak_Jones_1995}, coproducts~\cite{Ghani_Uustalu_2004,Adamek_Milius_Bowler_Levy_2012} and tensors \cite{Hyland_Plotkin_Power_2006,Power_2005}.  
The main approach for studying compositions of two monads with functor parts $M$ and $T$ into a monad with functor part $MT$ is via distributive laws \cite{Beck_1969}.
However, distributive laws do not always exist \cite{Klin_Salamanca_2018,Varacca_Winskel_2006,Zwart_Marsden_2019}.
These negative results motivated investigations into weaker notions of distributive laws \cite{Bohm_2010,Garner_2020,Street_2009}.
For some of the most prominent examples, including the lack of a distributive law of the probability distributions monad over the powerset monad \cite{Varacca_Winskel_2006}, the failure was located at one of the two unit axioms.
To overcome this, Garner~\cite{Garner_2020} defines a notion of \emph{weak distributive law} which drops the problematic axiom.
The usefulness of this concept has been demonstrated as
several monads have been proven to result from weak distributive laws: the Vietoris monad \cite{Garner_2020}, the convex powerset monad \cite{Goy_Petrisan_2020} and the monad of finitely generated convex subsets \cite{Bonchi_Santamaria_2021}.

Following this line of research, Petri\c{s}an~\&~Sarkis~\cite{Petrisan_Sarkis_2021} defined the \emph{semifree monad} for a monad $M$ on the functor coproduct $M\semifree \defeq \Id + M$.
They demonstrated a one-to-one correspondence between weak distributive laws with $M$ and distributive laws with $M\semifree$ satisfying an extra condition.
To achieve this, they first proved the existence of an isomorphism of categories between $M\semifree$-algebras and $M$-semialgebras, the latter being $M$-algebras that only satisfy the associativity axiom, but not necessarily the unit axiom.
A similar result was proved by Hyland and Tasson \cite[Proposition 27]{Hyland_2021} in the context of $2$-monads and $2$-categories.

The algebraic presentation of finitary monads by algebraic theories allows for equational reasoning about programs with computational effects, and is currently an active area of research, see e.g.~\cite{Bonchi_Santamaria_2021,Bonchi_Sokolova_Vignudelli_2019,Mio_Sarkis_Vignudelli_2021,Mio_Vignudelli_2020,Zwart_Marsden_2019}.

The main contribution of this paper is to show that an algebraic presentation of the semifree monad $M\semifree$ can be obtained uniformly from an algebraic presentation of $M$ (\cref{thm:conjecture}).
This uniform presentation was conjectured 
in \cite{Petrisan_Sarkis_2021}.
We apply the result to obtain algebraic presentations of the semifree monad over the exception monad, the list monad, the multiset monad, the finite powerset monad and the state monad.

We give a brief sketch of the proof.
Given an algebraic theory $(\Sigma,E)$, we first note that it suffices to prove the result for the free monad $T$ of the theory, which is clearly presented by it.
Let $(\Sigma\semifree, E\semifree)$ be the proposed presentation of $T\semifree$, where the signature $\Sigma\semifree$ consists of all operations in $\Sigma$ plus a new unary, idempotent operation denoted $\newa$. 
Due to the isomorphism between $T\semifree$-algebras and $T$-semialgebras it suffices to show an isomorphism between the categories of $T$-semialgebras and $(\Sigma\semifree,E\semifree)$-algebras.  

\begin{itemize}
    \item
    Given a $T$-semialgebra $(X,\alpha)$, we obtain a $(\Sigma\semifree,E\semifree)$-algebra by interpreting the new unary operation $\newa$ by embedding elements of $X$ into $TX$ using the unit $\eta_X$ of $T$ and then applying $\alpha$.
    The old operations from $\Sigma$ are interpreted similarly.
    We then verify by inductive arguments that this algebra satisfies all equations in $E\semifree$.
    
    \item
    Given a $(\Sigma\semifree,E\semifree)$-algebra, we obtain an $T$-semialgebra $(X,\alpha)$ using that elements of $TX$ are congruence classes of terms. 
    This allows us to show that $\alpha$ is indeed associative by an induction on the term structure of the elements of $TTX$.
\end{itemize}
Apart from providing equational reasoning for semifree monads, this uniform presentation result provides a means to study weak distributive laws via the presentations of the monads involved, using a similar approach as
Zwart~\&~Marsden~\cite{Zwart_Marsden_2019}. 
We discuss their work below and in \cref{sec:future_work}.


The paper is organised as follows.
\cref{sec:preliminaries} introduces preliminary definitions of monads, universal algebra, and semifree monads.
\cref{sec:theorem} states and proves the main result, \cref{thm:conjecture}.
\cref{sec:examples} presents examples of the application of \cref{thm:conjecture}.
\cref{sec:ideal_monad} investigates the relationship between the semifree construction and the notions of ideal monad and monad transformer. 
\cref{sec:future_work} concludes and discusses future work.
Omitted proofs can be found in the appendix.

\subsubsection{Related work:}
The work in the present paper concerns the question of how to obtain a monad presentation from existing ones, in a uniform manner.

Zwart~\&~Marsden~\cite{Zwart_Marsden_2019} use presentations of monads to give general results about the non-existence of distributive laws (``no-go theorems''). This approach allowed them to answer several open questions, including the 50 year-old conjecture by Beck \cite[Example 4.1]{Beck_1969} that addition cannot distribute over multiplication. On the positive side, when a distributive law exists, they show how to obtain a presentation of the composite monad from the monads and the distributive law.

Monads of the form $\mathcal{C}+1$ and other modifications of $\mathcal{C}$, the monad of convex subsets of distributions are studied in \cite{Mio_Sarkis_Vignudelli_2021}.
Both positive and negative results are obtained in different categories.
No uniform presentation result such as \cref{thm:conjecture} is given.

Recent work on algebraic presentations of specific monads include a presentation of the monad $\mathcal{C}$ \cite{Bonchi_Sokolova_Vignudelli_2019},
and presentations for monads arising via weak distributive laws that combine the monads of semilattices and semimodules \cite{Bonchi_Santamaria_2021}.

\subsubsection{Acknowledgments:} We thank Ralph Sarkis, and Roy Overbeek for useful discussion, suggestions and corrections. 
We also thank all anonymous reviewers for their valuable feedback and suggestions.
Alo\"is Rosset and J\"org Endrullis received funding from the Netherlands Organization for Scientific Research (NWO) under the Innovational Research Incentives Scheme Vidi (project. No. VI.Vidi.192.004).

\section{Preliminaries}\label{sec:preliminaries}

We assume the reader is familiar with basic notions of category theory 
\cite{Awodey_2006,MacLane_1971,Riehl_2017}.
In this section, we recall basic definitions and fix notation concerning monads, algebraic theories and presentations.
We also recall basic definitions and results of semifree monads.

\begin{definition}
    A \myemph{monad} $(M,\eta,\mu)$ on a category $\cat{C}$ is a triple consisting of an endofunctor $M: \cat{C} \to \cat{C}$, and two natural transformations, the \myemph{unit} $\eta : \Id \Rightarrow M$ and the \myemph{multiplication} $\mu : M^2 \Rightarrow M$ that make \eqref{eqn:unit_axioms} and \eqref{eqn:associativity_axiom} commute.
    We refer to \eqref{eqn:associativity_axiom} as the associativity of $\mu$.
    
    \noindent\begin{tabularx}{\textwidth}{@{}XX@{}}
        \begin{equation}
            \label{eqn:unit_axioms}
            \begin{tikzcd}[ampersand replacement=\&]
                M \ar[r, "M\eta"] \ar[rd, equal] \& M^2 \ar[d, "\mu"] \& M \ar[l, "\eta M"'] \ar[ld, equal] \\
                \& M \&
            \end{tikzcd}
        \end{equation} &
        \begin{equation}
            \label{eqn:associativity_axiom}
            \begin{tikzcd}[ampersand replacement=\&]
                M^3 \ar[d, "M\mu"'] \ar[r, "\mu M"] \& M^2 \ar[d, "\mu"] \\
                M^2 \ar[r, "\mu"'] \& M
            \end{tikzcd}
        \end{equation}
    \end{tabularx}
\end{definition}

For convenience, we often refer to a monad $(M,\eta,\mu)$ by its functor part $M$.

\begin{definition}
    Given two monads $(M,\eta^M,\mu^M)$ and $(T,\eta^T,\mu^T)$ on a category $\cat{C}$, a \myemph{monad morphism} from $M$ to $T$ is a natural transformation $\sigma : M \Rightarrow T$ that makes \eqref{eqn:monad_map_axiom1} and \eqref{eqn:monad_map_axiom2} commute, where $\sigma \sigma \defeq \sigma_{T} \circ M\sigma = T\sigma \circ \sigma_M$.
    
    \noindent\begin{tabularx}{\textwidth}{@{}XX@{}}
        \begin{equation}
            \label{eqn:monad_map_axiom1}
            \begin{tikzcd}[row sep = 0.25em, ampersand replacement=\&]
                \& M \ar[dd, "\sigma"] \\
                \Id \ar[ur, "\eta^M"] \ar[dr, "\eta^T"'] \& \\
                \& T
            \end{tikzcd}
        \end{equation} &
        \begin{equation}
            \label{eqn:monad_map_axiom2}
            \begin{tikzcd}[ampersand replacement=\&]
                M^2 \ar[d, "\mu^M"'] \ar[r, "\sigma \sigma"] \& T^2 \ar[d, "\mu^T"] \\
                M \ar[r, "\sigma"'] \& T
            \end{tikzcd}
        \end{equation}
    \end{tabularx}
    If each component of $\sigma$ is an isomorphism, we say that the two monads are \myemph{isomorphic}.
    The category of monads on a category $\cat{C}$ and monad morphisms is denoted $\catmon{C}$.
\end{definition}

\begin{definition}
    Let $(M,\eta,\mu)$ be a monad on category $\cat{C}$.
    An (Eilenberg-Moore) $M$-\myemph{algebra} is a $\cat{C}$-morphism $\alpha : MX \to X$ for some $X \in \cat{C}$, denoted $(X,\alpha)$ for short, such that \eqref{eqn:unit_axioms2} and \eqref{eqn:associativity_axiom2} commute. 
    An $M$-\myemph{semialgebra} is a $\cat{C}$-morphism $\alpha : MX \to X$ that satisfies 
\eqref{eqn:associativity_axiom2}.

    \noindent\begin{tabularx}{\textwidth}{@{}XX@{}}
        \begin{equation}
            \label{eqn:unit_axioms2}
            \begin{tikzcd}[ampersand replacement=\&]
                X \ar[r, "\eta_X"] \ar[rd, equal] \& MX \ar[d, "\alpha"] \\
                \& X
            \end{tikzcd}
        \end{equation} &
        \begin{equation}
            \label{eqn:associativity_axiom2}
            \begin{tikzcd}[ampersand replacement=\&]
                M^2X \ar[d, "M\alpha"'] \ar[r, "\mu_X"] \& MX \ar[d, "\alpha"] \\
                MX \ar[r, "\alpha"'] \& X
            \end{tikzcd}
        \end{equation}
    \end{tabularx}
    An $M$-\myemph{(semi)algebra homomorphism} $f : (X,\alpha) \to (Y,\beta)$ between two $M$-(semi)algebras is a function $f : X \to Y$ such that the following diagram commutes:
    \begin{center}
        \begin{tikzcd}
            MX \ar[d, "\alpha"'] \ar[r, "Mf"] & MY \ar[d, "\beta"] \\
            X \ar[r, "f"'] & Y
        \end{tikzcd}
    \end{center}
    The category of $M$-algebras and $M$-algebra homomorphisms is denoted $\EM{M}$ and called the \myemph{Eilenberg-Moore category} of $M$.
    The category $\EMs{M}$ consists of $M$-semialgebras and $M$-semialgebra homomorphisms.
\end{definition}

    We denote the {\bfseries{coproduct}} of two objects $X$ and $Y$ in a category $\cat{C}$ by $X + Y$, the left and right injections by $\inl^{X+Y} : X \to X+Y$ and $\inr^{X+Y} : Y \to X+Y$, and the arrow given by the universal property of the coproduct for arbitrary $f:X \to Z$ and $g : Y \to Z$ by $[f,g] : X + Y \to Z$.

\noindent We now recall a few basic notions of universal algebra.

\begin{definition}
    \begin{itemize}
        \item
        An \myemph{ algebraic signature} $\Sigma$ is a set of operation symbols each having its own arity $n \in \N$, denoted $(\op : n)$ for an $n$-ary $\op \in \Sigma$.
        
        \item 
        Given an algebraic signature $\Sigma$ and a set $X$, the set $\terms{\Sigma}{X}$ of
        $\Sigma$-terms over $X$ is defined inductively as follows.
        Elements in $X$ are terms, and they are said to be of depth zero, written $\cp(v) = 0$.
        If $t_1,\ldots,t_n$ are terms, and $(\op : n) \in \Sigma$, then $t \defeq \op(t_1,\ldots,t_n)$ is a term of depth $\cp(t) \defeq \mathsf{max}\set{\cp(t_1),\ldots,\cp(t_n)} + 1$.
        We define constants (i.e., nullary operations) to have depth 1.
        
        \item
        We fix a set $\variables = \set{v_1, v_2, v_3, \ldots}$ of variables.
        To indicate that the variables appearing in $t \in \terms{\Sigma}{\variables}$ are in the set $\set{v_1,\ldots,v_n}$, we write $t(v_1,\ldots,v_n)$.
        
        \item 
        A $\Sigma$-\myemph{algebra} is a pair $(X, \Brackets{\cdot})$, where $X$ is a set and $\Brackets{\cdot}$ is a collection of interpretations: for each $(\op : n) \in \Sigma$, we have $\Brackets{\op} : X^n \to X$.
        
        \item
        Given a $\Sigma$-algebra $(X, \Brackets{\cdot})$, any function $f : Y \to X$ extends to a unique homomorphism,
        $\Brackets{\cdot}_f : \terms{\Sigma}{Y} \to X$:
        \begin{align}
            \Brackets{y}_f &\defeq f(y), \text{ and} \label{eqn:def_[[]]_variables} \\
            \Brackets{\op(t_1,\ldots,t_m)}_f &\defeq \Brackets{\op} ( \Brackets{t_1}_f, \ldots, \Brackets{t_m}_f). \label{eqn:def_[[]]_operations}
        \end{align}
        When $f$ is the identity $\id_X$, the subscript is omitted.
        The function $f$ is often a variable assignment $\sigma : \variables \to X$ to obtain an interpretation of $\terms{\Sigma}{\variables}$.
        
        \item
        An \myemph{equation} over $\Sigma$ is a pair of terms $(s,t) \in \terms{\Sigma}{\variables} \times \terms{\Sigma}{\variables}$.
        
        \item
        An \myemph{algebraic theory} is a
        pair $(\Sigma,E)$ consisting of an algebraic signature $\Sigma$ and set $E$ of equations over $\Sigma$.

        \item
        A $(\Sigma,E)$-\myemph{algebra} is a $\Sigma$-algebra $(X,\Brackets{\cdot})$ for which the interpretation satisfies all equations in $E$, meaning that for all $(s,t) \in E$ and all variable assignments $\sigma$, $\Brackets{s}_\sigma = \Brackets{t}_\sigma$.
        
        \item 
        A $(\Sigma,E)$-\myemph{algebra homomorphism} between two $(\Sigma,E)$-algebras $(X,\Brackets{\cdot})$ and $(X', \Brackets{\cdot}')$ is a function $f : X \to X'$ such that $f \circ \Brackets{\op} = \Brackets{\op}' \circ f^n$ for all $(\op : n) \in \Sigma$.
        
        \item 
        The category $\catalg{\Sigma,E}$ consists of $(\Sigma,E)$-algebras and $(\Sigma,E)$-algebra homomorphisms.
        
        \item
        Given terms $s$ and $t$, we write $E \vdash s = t$  to denote that $s=t$ is derivable from $E$ in equational logic,
        and $E \vDash s=t$ to denote that the equation $s=t$ holds in all $(\Sigma,E)$-algebras.
        
        \item
        The \myemph{free $(\Sigma,E)$-algebra} on a set $X$ is the $(\Sigma,E)$-algebra $(\freealgebra{\Sigma}{X}{E},\freeinterX{\cdot})$ with carrier set consisting of $\terms{\Sigma}{X}$ modulo the smallest congruence relation containing $E_{}$. 
        The congruence class of a term $t$ is denoted $\ov{t}$.
        The interpretation of an operation $(\op : m) \in \Sigma$ is
        \[
            \freeinterX{\op}(\ov{t_1}, \ldots, \ov{t_m}) \defeq \ov{\op(t_1,\ldots,t_m)}.
        \]
        
        \item
        The \myemph{free functor} $F: \catset \to \catalg{\Sigma,E}$ sends a set $X$ to the free $(\Sigma,E)$-algebra $(\freealgebra{\Sigma}{X}{E},\Brackets{\cdot})$.
        The adjective "free" is also true in the categorical sense, i.e., we have a free-forgetful adjunction
        \begin{center}
            \begin{tikzcd}
                F : \catset \ar[r, shift left=1.3, ""] \ar[r, phantom, "\scriptscriptstyle\bot"] & \catalg{\Sigma,E} : U  \ar[l, shift left=1.3,""]
            \end{tikzcd}
        \end{center}
        The composite $\freemonad{\Sigma,E} \defeq UF$ is a monad (see e.g.~\cite[VI.1]{MacLane_1971}).
        Its unit $\freemonadunit$ sends an element to its equivalence class $x \mapsto \ov{x}$ and its multiplication $\freemonadmult $ flattens terms $\ov{t[\ov{t_i}/v_i]} \mapsto \ov{t[t_i/v_i]}$.
    \end{itemize}
\end{definition}

We can finally define the central concept of algebraic presentation.
\begin{definition} \label{def:algebraic_presentation}
    An algebraic theory $(\Sigma,E)$ is an \myemph{algebraic presentation} of a $\catset$-monad $(M,\eta,\mu)$ if $(\freemonad{\Sigma,E},\freemonadunit,\freemonadmult) \isom (M,\eta,\mu)$.
\end{definition}

Note that a monad can have multiple presentations.
From the definition, we immediately have the trivial example that an algebraic theory $(\Sigma,E)$ is an algebraic presentation of its free monad $\freemonad{\Sigma,E}$.

The class of monads that admit presentations is precisely the class of \emph{finitary monads}; for the definition and also the proof of this correspondence see e.g.~\cite[Section 3.18, p.149]{Adamek_Rosicky_1994}.
We therefore work only with finitary monads in the article.

Given a $\catset$-monad $M$ and an algebraic theory $(\Sigma,E)$,  the categories of algebras $\EM{M}$ and $\catalg{\Sigma,E}$ are concrete categories, and in this paper, it turns out to be more convenient to work with an equivalent definition of algebraic presentation formulated in terms of concrete isomorphisms. 

\begin{definition} \label{def:concrete}
    A category $\cat{C}$ is \myemph{concrete} if there is a faithful functor $U : \cat{C} \to \catset$, usually a forgetful functor.
    A functor $F : \cat{C} \to \cat{D}$ between concrete categories is itself \myemph{concrete} if it commutes with the faithful functors $U_\cat{D} \circ F = U_\cat{C}$.
    We write $ \cat{C} \cisom \cat{D}$ to denote that categories $\cat{C}$ and $\cat{D}$ are concretely isomorphic.
\end{definition}

The following lemma is well-known and is a direct consequence of e.g.~\cite[Theorem III.6.3]{Barr_Wells_1985_TTT}.

\begin{lemma}
    \label{lem:isom_monads_iff_isom_algebras_cats}
    For $\catset$-monads $(M,\eta^M,\mu^M),(T,\eta^T,\mu^T)$, we have
    \[
        (M,\eta^M,\mu^M) \isom (T,\eta^T,\mu^T) 
        \quad \Longleftrightarrow \quad
        \EM{M} \cisom \EM{T}.
    \]
\end{lemma}

It gives the following alternative formulation of algebraic presentation.

\begin{lemma}\label{lem:alg-pres-equiv}
    An algebraic theory $(\Sigma,E)$ is an algebraic presentation of a (finitary) monad $(M,\eta,\mu)$ if and only if $\EM{M} \cisom \catalg{\Sigma,E}$.
\end{lemma}

\begin{proof}
    Since $\catalg{\Sigma,E} \cisom \EM{\freemonad{\Sigma,E}}$ (see e.g.~\cite[VI.8.1]{MacLane_1971}), 
    the result follows immediately from Lemma~\ref{lem:isom_monads_iff_isom_algebras_cats}.     
\end{proof}

\begin{remark}\label{rem:concrete}
    In the literature, the definition of algebraic presentation is often stated as the condition $\EM{M} \isom \catalg{\Sigma,E}$, i.e., it leaves the ``concrete'' part implicit, see e.g. \cite{Bonchi_Santamaria_2021,Bonchi_Sokolova_Vignudelli_2019,Mio_Vignudelli_2020,Mio_Sarkis_Vignudelli_2021,Petrisan_Sarkis_2021}. The ``concrete'' part is not necessary in these papers, since they establish algebraic presentations (by indeed establishing concrete isomorphisms), but they do not prove results that assume the existence of an algebraic presentation. In the present paper, we assume that a presentation for $M$ is given, and establish one for $M^s$, and the proof requires the isomorphism $\EM{M} \cisom \catalg{\Sigma,E}$ to be concrete. 
\end{remark}


The following lemma states two well-known facts that we will need later in proofs.

\begin{lemma} \label{lem:homom_by_nat_of_mu}
    Let $(\Sigma,E)$ be an algebraic theory and $T$ denote its free monad $\freemonad{\Sigma,E}$.
    Given a function $f:Y \to X$, then the following are $(\Sigma,E)$-algebra homomorphisms:
    \begin{align*}
        \mu_X &: (TTX, \freeinterTX{\cdot}) \to (TX, \freeinterX{\cdot}), \\
        Tf &: (TY, \freeinterY{\cdot}) \to (TX, \freeinterX{\cdot}).
    \end{align*}
\end{lemma}

The semifree monad $M\semifree$ for a monad $M$ was introduced in \cite{Petrisan_Sarkis_2021} by Petri\c{s}an and Sarkis.

\begin{definition}[\cite{Petrisan_Sarkis_2021}] \label{def:semifree_monad}
    Given a monad $(M,\eta,\mu)$ on a category $\cat{C}$ having all finite coproducts, the \myemph{semifree monad} on $M$ is a monad $(M\semifree,\eta\semifree,\mu\semifree)$, where
    \begin{align*}
        M\semifree &\defeq \Id_C + M, \\
        \eta\semifree &\defeq \inl^{\Id + M}, \\
        \mu\semifree &\defeq [\Id_{\Id+M}, \inr^{\Id+M} \circ \mu \circ M[\eta, \Id_M] ].
    \end{align*}
\end{definition}

Note that the unit $\eta\semifree$ injects a set $X$ to its copy on the left in $X + MX$.
Petri\c{s}an and Sarkis showed in Theorems~3.4 and 4.3 of \cite{Petrisan_Sarkis_2021} that:
\begin{itemize}
    \item 
    There is a (concrete) isomorphism $\EM{M\semifree} \cisom \EMs{M}$.
    
    \item
    There is a bijection between weak distributive laws $\lambda: MT \Rightarrow TM$ and distributive laws $\delta : M\semifree T \Rightarrow TM\semifree$ satisfying an extra condition.
\end{itemize}

The semifree construction takes a monad as input and outputs another monad.
The semifree construction can be made into a functor on $\catmon{C}$ as follows. Given a monad morphism $\sigma : M \Rightarrow T$, and a set $X$, let
\begin{equation}\label{eq:semifree-functor-action-bis}
    \sigma\semifree_X \defeq \Big( X + MX \xrightarrow{\id_X + \sigma_X } X + TX \Big).
\end{equation}

\begin{lemma} \label{lem:semifree_construction_is_a_functor}
    The mapping $(-)\semifree : \catmon{C} \to \catmon{C}$ is a functor.
\end{lemma}

Here is another quick observation that will be needed later.

\begin{lemma} \label{lem:monads_isomorphic_imply_semifree_monads_isomorphic}
    Take two monads $(M,\eta^M,\mu^M), (T,\eta^T,\mu^T)$ on a category $\cat{C}$ that has all finite coproducts.
    If they are isomorphic $M \isom T$, then their respective semifree monads are also isomorphic $M\semifree \isom T\semifree$.
\end{lemma}

\begin{proof}
    Suppose we have a monad isomorphism $\sigma : M \Rightarrow T$.
    We prove $\sigma\semifree : M\semifree \isom T\semifree$ is a monad isomorphism.
    Since $\sigma$ is an isomorphism, it has an inverse $\tau : T \Rightarrow M$.
    We show that $\tau\semifree$ is the inverse of $\sigma\semifree$. For any set $X$,
    \begin{align*}
        \sigma\semifree_X \circ \tau\semifree_X
        &= (\id_X + \sigma_X) \circ (\id_X + \tau_X) 
        \tag{def \eqref{eq:semifree-functor-action-bis}} \\
        &= \id_X + (\sigma_X \tau_X)
        \tag{coproduct property} \\
        &= \id_X + \id_{TX}
        \tag{$\sigma,\tau$ inverses} \\
        &= \id_{T\semifree_X}
    \end{align*}
    and similarly $\tau\semifree_X \circ \sigma\semifree_X = \id_{M\semifree_X}$, concluding the proof.
\end{proof}

\newpage

\section{Algebraic Presentation of Semifree Monads}\label{sec:theorem}

In this section, we state and prove the main result of the paper.
We prove that given an algebraic presentation of a (finitary) $\catset$-monad $(M,\eta,\mu)$, we can derive an algebraic presentation of the semifree monad $(M\semifree,\eta\semifree,\mu\semifree)$.
In particular, if $M$ is a finitary monad, then its semifree monad $M\semifree$ is finitary too.

Before we state the theorem, we give some intuitions for the presentation of $(M\semifree,\eta\semifree,\mu\semifree)$.
Recall that $M\semifree = X + MX$ and the unit $\eta\semifree_X : X \to X + MX$ is the left injection.
In terms of presentation, this means that the left copy of $X$ becomes the ``new'' set of variables.
As a consequence, the ``old'' set of variables $\eta_X(X) \subseteq MX$ is now free in $M\semifree$ of the constraints, such as the unit laws \eqref{eqn:unit_axioms} and \eqref{eqn:unit_axioms2}, that it had in $M$.
The inclusion of $X$ via $\eta_X$ in $M\semifree X$ corresponds to a new unary operation $(\newa : 1)$ in the presentation of $M\semifree$.
On the semantic level, suppose we have an $M\semifree$-algebra $\gamma : M\semifree X \to X$.
By Theorem 3.4 in \cite{Petrisan_Sarkis_2021} and by looking at its proof, we see that $\gamma$ must be of the form  $[\Id_X,\alpha]$ where $\alpha : MX \to X$ is an $M$-semialgebra.
Notice the following:
\begin{equation} \label{eqn:a_etaX_a=a}
    \alpha \circ \eta_X \circ \alpha 
    \stackrel{\eta \text{ nat.}}{=}
    \alpha \circ M\alpha \circ \eta_{MX}
    \stackrel{\eqref{eqn:associativity_axiom2}}{=}
    \alpha \circ \mu_X \circ \eta_{MX}
    \stackrel{\eqref{eqn:unit_axioms}}{=}
    \alpha.
\end{equation}
Hence, $\alpha \circ \eta_X$ is an idempotent.
This map will be our choice for the interpretation of the new symbol $(\newa : 1)$.

\begin{definition}\label{def:Sigmas_Es}
    Given an algebraic theory $(\Sigma,E)$, we define a new algebraic theory $(\Sigma\semifree,E\semifree)$ by $\Sigma_{}\semifree \defeq \Sigma_{} \uplus \set{\newa:1}$ and
    $E_{}\semifree$ containing the following:
    \begin{align}
        \newa \newa v_1 &= \newa v_1, \label{eqn:idempotency} \\
        \newa (\op(v_1, \ldots, v_n)) &= \op(v_1,\ldots,v_n), \label{eqn:a_in_front_disappear} \\
        \op(\newa v_1, \ldots, \newa v_n) &= \op(v_1, \ldots, v_n), \label{eqn:a_inside_disappear} \\
        t(\newa v_1, \ldots, \newa v_n) &= s(\newa v_1, \ldots, \newa v_n), \label{eqn:a_inside_terms_disappear}
    \end{align}
    for all $(\op : n) \in \Sigma_{}$ and $\left(t(v_1,\ldots,v_n)=s(v_1,\ldots,v_n)\right) \in E_{}$.
\end{definition}

We have the trivial fact that an algebraic theory $(\Sigma,E)$ is an algebraic presentation of its free monad $\freemonad{\Sigma,E}$.
The next theorem states that the algebraic theory $(\Sigma\semifree,E\semifree)$ of \cref{def:Sigmas_Es} is an algebraic presentation of $\freemonad{\Sigma,E}\semifree$, the semifree monad on $\freemonad{\Sigma,E}$.

\begin{theorem} \label{thm:big_theorem}
    $(\Sigma\semifree, E\semifree)$ is an algebraic presentation of $(\freemonad{\Sigma,E}\semifree, \freemonadunit\semifree, \freemonadmult\semifree)$.
\end{theorem}

The proof of this theorem is the goal of the rest of \cref{sec:theorem}.
As a direct consequence, we have the following corollary.
It was originally formulated as Conjecture 5.4 in \cite{Petrisan_Sarkis_2021} by Petri\c{s}an and Sarkis.

\begin{corollary} \label{thm:conjecture}
    If $(\Sigma,E)$ is an algebraic presentation of a monad $(M,\eta,\mu)$,
    then $(\Sigma\semifree,E\semifree)$ is an algebraic presentation of $(M\semifree,\eta\semifree,\mu\semifree)$, the semifree monad on $M$. 
\end{corollary}

\begin{proof}
    Assume that we have a monad isomorphism $\freemonad{\Sigma,E} \isom M$.
    By \cref{lem:monads_isomorphic_imply_semifree_monads_isomorphic}, their semifree monads are also isomorphic $\freemonad{\Sigma,E}\semifree \isom M\semifree$.
    By \cref{thm:big_theorem}, $\freemonad{\Sigma,E}\semifree \isom \freemonad{\Sigma\semifree,E\semifree}$.
    Hence, $M\semifree \isom \freemonad{\Sigma\semifree,E\semifree}$, which means that $(\Sigma\semifree,E\semifree)$ is an algebraic presentation of $M\semifree$.
\end{proof}

We will need a few technical lemmas.
The next one shows that equations $\eqref{eqn:a_in_front_disappear}$ and $\eqref{eqn:a_inside_disappear}$ extend inductively to all terms of depth at least $1$.

\begin{lemma} \label{lem:generalisation_eq_23}
    For all $t \in \terms{\Sigma}{\variables}$ of depth at least $1$,
    an $\Sigma_{}\semifree$-algebra that satisfies \eqref{eqn:idempotency}-\eqref{eqn:a_inside_disappear} also satisfies 
    \begin{align} 
          \newa t(v_1,\ldots,v_n) &= t(v_1,\ldots,v_n), \text{ and}
          \label{eqn:generalisation2} \\
          t(\newa v_1, \ldots, \newa v_n) &= t(v_1,\ldots,v_n).
          \label{eqn:generalisation3}
    \end{align}
\end{lemma}

In \cref{def:Sigmas_Es}, equation \eqref{eqn:a_inside_terms_disappear} tells us that
equations from $E_{}$ give rise to equations in $E\semifree_{}$ by substituting $v \mapsto \newa v$, for all variables $v$.
The next lemma indicates that the same procedure can be done for theorems, i.e., that a theorem deducible from $E_{}$ becomes, after the substitution $v \mapsto \newa v$, a theorem deducible from $E_{}\semifree$.

\begin{lemma} \label{lem:deduction_tree_adaptation}
    For all terms $t(v_1,\ldots,v_n),s(v_1,\ldots,v_n) \in \terms{\Sigma}{\variables}$,
    \[
        E_{} \vdash t(v_1,\ldots,v_n) = s(v_1,\ldots,v_n)
        \ \Longrightarrow \ 
        E_{}\semifree \vdash t(\newa v_1,\ldots, \newa v_n) = s(\newa v_1,\ldots,\newa v_n).
    \]
\end{lemma}

The last lemma is purely technical.
Its reasoning appears multiple times in different proofs.
Stating it here allows us to prove it once for all.

\begin{lemma} \label{lem:formula_a_opX}
    Let $(\Sigma,E)$ be an algebraic theory with free monad $(\freemonad{\Sigma,E},\freemonadunit,\freemonadmult)$.
    For every $\freemonad{\Sigma,E}$-semialgebra $\alpha: \freemonad{\Sigma,E} X \to X$ and operation symbol $(\op:n) \in \Sigma$, we have
    \begin{equation} \label{eqn:formula_a_opX}
        \alpha \circ \freeinterX{\op} = \alpha \circ \freeinterX{\op} \circ \eta_X^n \circ \alpha^n.
    \end{equation}
\end{lemma}

We now tackle the proof of \cref{thm:big_theorem}.
For simplicity, we will denote in the rest of this section the free monad $\smash{(\freemonad{\Sigma,E},\freemonadunit,\freemonadmult)}$ simply by $\smash{(T,\eta,\mu)}$.
To prove that $(\Sigma\semifree,E\semifree)$ is an algebraic presentation of $T\semifree$, it suffices by \cref{lem:alg-pres-equiv} to prove that $\catalg{\Sigma\semifree,E\semifree} \cisom \EM{T\semifree}$.
Recall that $\EM{T\semifree} \cisom \EMs{T}$, by Theorem~3.4 in \cite{Petrisan_Sarkis_2021}, i.e., $T\semifree$-algebras are concretely isomorphic to $T$-semialgebras.
Therefore, it suffices to prove that $\EMs{T} \cisom \catalg{\Sigma\semifree, E\semifree}$.

\subsection{From $T$-semialgebras to $(\Sigma_{}\semifree, E_{}\semifree)$-algebras} \label{sec:forward_direction}
    
In the forward direction, suppose we have a semialgebra $\alpha: TX \to X$.
We want to obtain an $(\Sigma_{}\semifree,E_{}\semifree)$-algebra.
It will be constructed with carrier $X$, since we are aiming for a concrete isomorphism. 

\begin{definition} \label{def:forward_functor}
    We define the mapping
    \[
        G: \EMs{T} \to \catalg{\Sigma_{}\semifree, E_{}\semifree} \\
    \]
    by $G(X,\alpha) \defeq (X, \Parens{\cdot})$ on objects, where the interpretation $\Parens{\cdot}$ is defined on the operation symbols $\newa : 1$ and $(\op : n) \in \Sigma_{}$ as
    \begin{align}
        \Parens{\newa} &\defeq \left( X \xrightarrow{\eta_X} TX \xrightarrow{\alpha} X \right) \label{eqn:def_newa_interpretation}, \\
        \Parens{\op} &\defeq \left( X^n \xrightarrow{(\eta_X)^n} (TX)^n \xrightarrow{  \freeinterX{\op}} TX \xrightarrow{\alpha} X \right), \label{eqn:def_op_interpretation}
    \end{align}
    and $G(f) \defeq f$ on morphisms.
\end{definition}

The goal now is to demonstrate that $G$ is well-defined on objects and arrows. It then follows immediately that $G$ is a functor due to being essentially identity on arrows.
To this end, we first establish in \cref{lem:generalisation_of_def}, a property that generalises both \eqref{eqn:def_newa_interpretation} and \eqref{eqn:def_op_interpretation} into one formula.
Then, we show in \cref{lem:G_well_defined_equations_satisfied} that $G$ indeed outputs $(\Sigma\semifree_{},E\semifree_{})$-algebras.
Lastly, we show in \cref{lem:G_well_defined_functoriality} that $G$ outputs $(\Sigma\semifree,E\semifree)$-algebra homomorphisms, and hence is also well-defined on arrows.

\begin{lemma} \label{lem:generalisation_of_def}
    For all terms $t(v_1,\ldots,v_n) \in \terms{\Sigma}{\variables}$ of depth at least $1$, and all variable assignments $\sigma : \variables \to X$, we have
    \begin{equation} \label{eqn:generalisation_of_def}
       \Parens{t}_\sigma = \alpha \circ \freeinterX{t}_{\eta_X \circ \sigma}.
    \end{equation}
\end{lemma}

\begin{lemma} \label{lem:G_well_defined_equations_satisfied}
    For all $T$-semialgebras $\alpha: TX \to X, G(X,\alpha)$ is a $(\Sigma\semifree,E\semifree)$-algebra.
\end{lemma}

\begin{proof}
    We check that $(X,\Parens{\cdot}) \defeq G(X,\alpha)$ satisfies the equations in $E\semifree_{}$.
    Let $\sigma\colon \variables \to X$ be a variable assignment. 
    \begin{enumerate}[(i)]
        \item 
        For $E\semifree$-equations arising from \eqref{eqn:idempotency}, we have:
        \begin{align*}
            \Parens{\newa \newa v_1}_\sigma &= \Parens{\newa} \Parens{\newa} (\sigma v_1) \\
            &= (\alpha \circ \eta_X) \circ (\alpha \circ \eta_X) (\sigma v_1)
            \tag{by \eqref{eqn:def_newa_interpretation}} \\
            &= \alpha \circ \eta_X (\sigma v_1) 
            \tag{by \eqref{eqn:a_etaX_a=a}} \\
            &= \Parens{\newa v_1}_\sigma.
        \end{align*}
        
        \item
        For $E\semifree$-equations arising from \eqref{eqn:a_in_front_disappear}.
        \begin{align*}
            \Parens{\newa ( &\op (v_1, \ldots, v_n))}_\sigma  \\
            &= \Parens{\newa} \circ \Parens{\op} (\sigma v_1,\ldots,\sigma v_n) \\
            &= \alpha \circ \eta_X \circ \alpha \circ \freeinterX{\op} \circ (\eta_X)^n (\sigma v_1,\ldots,\sigma v_n) 
            \tag{by \eqref{eqn:def_newa_interpretation}, \eqref{eqn:def_op_interpretation}} \\
            &= \alpha \circ \freeinterX{\op} \circ (\eta_X)^n (\sigma v_1,\ldots,\sigma v_n)
            \tag{by \eqref{eqn:a_etaX_a=a}} \\ %
            &= \Parens{\op(v_1,\ldots,v_n)}_\sigma \tag{by \eqref{eqn:def_op_interpretation}}
        \end{align*}
        
        
        \item For $E\semifree$-equations arising from \eqref{eqn:a_inside_disappear}, we have:
        \begin{align*}
            \Parens{\op &(\newa v_1,  \ldots,\newa v_n)}_\sigma \\
            &= \Parens{\op} ( \Parens{\newa v_1}_\sigma, \ldots, \Parens{\newa v_n }_\sigma ) \\
            &= \alpha \circ \freeinterX{\op} \circ \eta_X^n (\alpha \circ \eta_X (\sigma v_1), \ldots, \alpha \circ \eta_X (\sigma v_n)) \tag{by \eqref{eqn:def_newa_interpretation}, \eqref{eqn:def_op_interpretation}} \\
            &= \alpha \circ \freeinterX{\op} \circ \eta_X^n \circ \alpha^n \circ \eta_X^n (\sigma v_1,\ldots,\sigma v_n) \\
            &= \alpha \circ \freeinterX{\op} \circ \eta_X^n (\sigma v_1,\ldots,\sigma v_n) \tag{by \eqref{eqn:formula_a_opX}} \\
            &= \Parens{\op(v_1,\ldots,v_n)}_\sigma. \tag{by \eqref{eqn:def_op_interpretation}} 
        \end{align*}
        
        \item For $E\semifree$-equations arising from \eqref{eqn:a_inside_terms_disappear}, let
        $(t(v_1,\ldots,v_n) = s(v_1,\ldots,v_n)) \in E_{}$.
        We have so far verified \eqref{eqn:idempotency}-\eqref{eqn:a_inside_disappear} and we can thus invoke \cref{lem:generalisation_eq_23}. 
        Since $(TX,\freeinterX{\cdot})$ is a $(\Sigma,E)$-algebra, and $\eta_X \circ \sigma : \variables \to TX$ is a variable assignement, we have $\freeinterX{t}_{\eta_X \circ \sigma} = \freeinterX{s}_{\eta_X \circ \sigma}$.
        We distinguish cases:
        \vspace{.1em}
        
        \begin{itemize}
            \item 
            Suppose $t$ and $s$ are variables, $v_1$ and $v_2$ respectively:
            \begin{align*}
                \freeinterX{v_1}_{\eta_X \circ \sigma} = \freeinterX{v_2}_{\eta_X \circ \sigma}
                &\Leftrightarrow \eta_X \circ \sigma(v_1) = \eta_X \circ \sigma(v_2) \tag{def \eqref{eqn:def_[[]]_variables}} \\
                &\Rightarrow \alpha \circ \eta_X \circ \sigma(v_1) = \alpha \circ \eta_X \circ \sigma(v_2) \\
                &\Rightarrow \Parens{\newa v_1}_\sigma = \Parens{ \newa v_2}_\sigma. \tag{def \eqref{eqn:def_newa_interpretation}}
            \end{align*}
            
            \item
            Suppose one is a variable and the other is not, say w.l.o.g.~that $t$ is $v_1$:
            \begin{align*}
                \freeinterX{v_1}_{\eta_X \circ \sigma} = \freeinterX{s}_{\eta_X \circ \sigma}
                &\Leftrightarrow \eta_X \circ \sigma(v_1) = \freeinterX{s}_{\eta_X \circ \sigma} 
                \tag{def \eqref{eqn:def_[[]]_variables}} \\
                &\Rightarrow \alpha \circ \eta_X \circ \sigma(v_1) = \alpha \circ \freeinterX{s}_{\eta_X \circ \sigma} \\
                &\Rightarrow \Parens{\newa v_1}_\sigma = \Parens{ s(v_1,\ldots,v_n) }_\sigma
                \tag{def \eqref{eqn:def_newa_interpretation}; \eqref{eqn:generalisation_of_def}} \\
                &\Rightarrow \Parens{\newa v_1}_\sigma = \Parens{ s(\newa v_1,\ldots, \newa v_n) }_\sigma. \tag{by \eqref{eqn:generalisation3}}
            \end{align*}
            
            \item
            Suppose neither of $t$ and $s$ is a variable:
            \begin{align*}
                \freeinterX{t}_{\eta_X \circ \sigma} \hspace{-.2em}= \freeinterX{s}_{\eta_X \circ \sigma} \hspace{-.2em}
                &\Rightarrow \alpha \circ \freeinterX{t}_{\eta_X \circ \sigma} = \alpha \circ \freeinterX{s}_{\eta_X \circ \sigma} \\
                &\Rightarrow \Parens{t(v_1,\ldots,v_n)}_\sigma = \Parens{s(v_1,\ldots,v_n)}_\sigma 
                \tag{by \eqref{eqn:generalisation_of_def}} \\
                &\Rightarrow \Parens{t(\newa v_1, \ldots, \newa v_n)}_\sigma = \Parens{s(\newa v_1, \ldots, \newa v_n)}_\sigma.
                \tag{by \eqref{eqn:generalisation3}}
            \end{align*}
        \end{itemize}
    \end{enumerate}
\end{proof}

\begin{lemma}\label{lem:G_well_defined_functoriality}
     $G$ maps $T$-semialgebra homomorphisms to $(\Sigma\semifree,E\semifree)$-algebra homomorphisms.
\end{lemma}

\begin{proof} 
    Suppose $f : (X,\alpha) \to (Y,\beta)$ is an $T$-semialgebra homomorphism.
    Let $(X, \Parens{\cdot}^X) \defeq G(X,\alpha)$ and $(Y, \Parens{\cdot}^Y) \defeq G(Y,\beta)$.
    We check that $G(f)$, which is defined as $f$ in \cref{def:forward_functor}, is an $(\Sigma\semifree_{}, E\semifree_{})$-algebra homomorphism, or in other words, that it commutes with the interpretations of the operations $\newa$ and each $(\op : n) \in \Sigma_{}$.
    \begin{center}
        \begin{tikzcd}
            X 
                \ar[d, "f"']
                \ar[r, "\eta_X"] 
                \arrow[urrd, to path= { 
                    -- ([yshift=1ex]\tikztostart.north) 
                    |- ([yshift=1ex]\tikztotarget.north) node[near end,above,font=\scriptsize]{${\Parens{\newa}^X}$}
                    -- (\tikztotarget)}]
                \ar[dr, phantom, "\scriptstyle (\eta \text{ nat.})"] 
            & TX
                \ar[r, "\alpha"]
                \ar[d, "Tf" description]
                \ar[dr, phantom, "\scriptstyle (f \text{ hom.})"]
            & X
                \ar[d, "f"] \\
            Y 
                \ar[r, "\eta_Y"'] 
                \arrow[drru, to path= { 
                    -- ([yshift=-1ex]\tikztostart.south) 
                    |- ([yshift=-1ex]\tikztotarget.south) node[near end, below, font=\scriptsize] {${\Parens{\newa}^Y}$}
                    -- (\tikztotarget)}]
            & TY
                \ar[r, "\beta"']
            & Y 
        \end{tikzcd}
    \qquad \qquad
        \begin{tikzcd}
            X^n
                \ar[d, "f^n"']
                \ar[r, "(\eta_X)^n"]
                \arrow[urrrd, to path= { 
                    -- ([yshift=2ex]\tikztostart.north) 
                    |- ([yshift=2ex]\tikztotarget.north) node[near end,above,font=\scriptsize]{$\Parens{\op}^X$}
                    -- (\tikztotarget)}]
                \ar[dr, phantom, "\scriptstyle (\eta \text{ nat.})^n"]
            & (TX)^n
                \ar[d, "(Tf)^n" description]
                \ar[r, "\freeinterX{\op}"]
                \ar[dr, phantom, "\scriptstyle (Lem. \ref{lem:homom_by_nat_of_mu})"]
            & TX
                \ar[d, "Tf" description]
                \ar[r, "\alpha"]
                \ar[dr, phantom, "\scriptstyle(f \text{ hom.})"]
            & X 
                \ar[d, "f"] \\
            Y^n 
                \ar[r, "(\eta_Y)^n"'] 
                \arrow[drrru, to path= { 
                    -- ([yshift=-2ex]\tikztostart.south) 
                    |- ([yshift=-2ex]\tikztotarget.south) node[near end, below, font=\scriptsize] {$\Parens{\op}^Y$}
                    -- (\tikztotarget)}]
            & (TY)^n 
                \ar[r, "\freeinterY{\op}"'] 
            & TY 
                \ar[r, "\beta"'] & Y
        \end{tikzcd}
    \end{center}
\end{proof}

\subsection{From $(\Sigma\semifree, E\semifree)$-algebras to $T$-semialgebras} \label{sec:backward_direction}

For the backward direction, given a $(\Sigma\semifree, E\semifree)$-algebra $(X, \Parens{\cdot})$, we want to define a $T$-semialgebra $\alpha:TX \to X$.
Since the elements of $TX$ are equivalence classes of terms, we can construct the desired semialgebra and the backward functor $H$ as follows.

\begin{definition}\label{def:backward_functor}
    We define the mapping
    \[
        H : \catalg{\Sigma_{}\semifree, E_{}\semifree} \to \EMs{T}
    \]
    by $H(X,\Parens{\cdot}) \defeq (X,\alpha)$ on objects, where $\alpha$ is defined as follows:
    \begin{equation} \label{eqn:def_semialgebra_a}
        \alpha : TX \to X : \ov{t} \mapsto \Parens{t}_{\Parens{\newa}},
    \end{equation}
    and $H(f) \defeq f$ on morphisms.
\end{definition}

We now show that $H$ is well-defined on objects and morphisms. It then follows immediately that $H$ is a functor due to being essentially identity on morphisms.
Since the definition of $\alpha$ relies on equivalence classes, we first show in \cref{lem:well_definedness_of_a} that $\alpha$ is well-defined, i.e., that changing representatives does not matter.
Then, we prove in \cref{lem:a_is_associative} that $\alpha$ is indeed a $T$-semialgebra.
Lastly, we show in \cref{lem:H_well_defined_functoriality} that $H$ outputs $T$-semialgebra homomorphisms.

\begin{lemma} \label{lem:well_definedness_of_a}
    Given a $(\Sigma\semifree, E\semifree)$-algebra $(X,\Parens{\cdot})$, and $(X,\alpha) \defeq H(X,\Parens{\cdot})$, then $\alpha$ from \eqref{eqn:def_semialgebra_a} is a well-defined function.
\end{lemma}


The next lemma states two identities that give more insight into how $\alpha$ works on distinct elements.
It makes future manipulations of $\alpha$ easier.

\begin{lemma}\label{lem:formula_a_opmx}
    Given a $(\Sigma\semifree, E\semifree)$-algebra $(X,\Parens{\cdot})$, and $(X,\alpha) \defeq H(X,\Parens{\cdot})$, then for all $x \in X$, $(\op:n \in \Sigma)$, and $c_1,\ldots,c_n \in \freealgebra{\Sigma}{X}{E}$, it holds that
    \begin{align}
        \alpha \circ \eta_X (x) &= \Parens{\newa} (x), \text{ and}\label{eqn:formula_a_eta} \\
        \alpha \circ \freeinterX{\op} (c_1,\ldots,c_n) &= \Parens{\op} (\alpha c_1, \ldots, \alpha c_n), \label{eqn:formula_a_opmx}
    \end{align}
\end{lemma}

The proof of the associativity of $\alpha$ highlights the use of term representatives.

\begin{lemma}\label{lem:a_is_associative}
    Given a $(\Sigma\semifree, E\semifree)$-algebra $(X,\Parens{\cdot})$, and $(X,\alpha) \defeq H(X,\Parens{\cdot})$, then $\alpha$ is a $T$-semialgebra, i.e., it satisfies the associativity axiom \eqref{eqn:associativity_axiom2}.
\end{lemma}

\begin{proof}[Proof of \cref{lem:a_is_associative}]
We have to show that ${\alpha \circ T\alpha (\ov{t}) = \alpha \circ \mu_X (\ov{t})}$ for all $\ov{t} \in TTX = \freealgebra{\Sigma}{TX}{E}$.
We do an induction on $t$:
    \begin{itemize}
        \item
        For the base case suppose $t$ is some $\ov{s} \in TX = \freealgebra{\Sigma}{X}{E}$.
        The goal can be reformulated as $\alpha \circ T\alpha \circ \eta_{TX} (\ov{s}) = \alpha \circ \mu_X \circ \eta_{TX} (\ov{s})$.
        By the unit law \eqref{eqn:unit_axioms}, it is the same as proving $\alpha \circ T\alpha \circ \eta_{TX} (\ov{s}) = \alpha (\ov{s})$.
        We distinguish cases for $s$:
        \begin{itemize}
            \item
            If $s$ is some $x \in X$, i.e., $\ov{s}=\eta_X(x)$:
            \begin{align*}
                \alpha \circ T\alpha \circ \eta_{TX} \circ \eta_X (x) &= \alpha \circ \eta_X \circ \alpha \circ \eta_X (x) \tag{$\eta$ nat.} \\
                &= \Parens{\newa} \circ \Parens{\newa}(x) \tag{by \eqref{eqn:formula_a_eta}} \\
                &= \Parens{\newa} (x) \tag{idempotence \eqref{eqn:idempotency}} \\
                &= \alpha \circ \eta_X (x) \tag{by \eqref{eqn:formula_a_eta}}
            \end{align*}
            
            \item
            If $s = \op(s_1,\ldots,s_m)$ for ${s_1,\ldots,s_n \in \terms{\Sigma}{X}}$ and $(\op:m)\in \Sigma$:
            \begin{align*} 
                \alpha \circ T\alpha \circ \eta_{TX} \big(\ov{s} \big) 
                &= \alpha \circ T\alpha \circ \eta_{TX} \big(\ov{\op(s_1,\ldots,s_m)} \big) 
                \\
                &= \alpha \circ T\alpha \circ \eta_{TX} \left( \freeinterX{\op} (\ov{s_1},\ldots,\ov{s_m}) \right)
                \tag{def.~$\freeinterX{\cdot}$}\\
                &= \alpha \circ \eta_X \circ \alpha \circ \freeinterX{\op} (\ov{s_1},\ldots,\ov{s_m})
                \tag{$\eta$ nat.} \\
                &= \Parens{\newa}  \circ \Parens{\op} (\alpha  (\ov{s_1}),\ldots,\alpha  (\ov{s_m}))
                \tag{by \eqref{eqn:formula_a_eta},\eqref{eqn:formula_a_opmx}} \\
                &= \Parens{\op} (\alpha  (\ov{s_1}),\ldots,\alpha  (\ov{s_m})) 
                \tag{equation \eqref{eqn:a_in_front_disappear}} \\
                &= \alpha \circ \freeinterX{\op} (\ov{s_1},\ldots,\ov{s_m}) 
                \tag{by \eqref{eqn:formula_a_opmx}} \\
                &= \alpha \big(\ov{\op(s_1,\ldots,s_m)} \big)
                \tag{def. $\Brackets{\cdot}$} \\
                &= \alpha \big(\ov{s} \big).
            \end{align*}
        \end{itemize}
        
        \item
        For the induction step of $t$, suppose $\alpha \circ T\alpha (\ov{t_i}) = \alpha \circ \mu_X (\ov{t_i})$ holds for $t_1,\ldots,t_n \in \terms{\Sigma}{TX}$ and let us prove it for $t = \op(t_1,\ldots,t_n)$.
        \begin{align*}
            \alpha \circ T\alpha \big( \ov{t} \big) 
            &= \alpha \circ T\alpha \big( \ov{\op(t_1,\ldots,t_n)} \big) 
            \\
            &= \alpha \circ T\alpha \circ \freeinterTX{\op} (\ov{t_1},\ldots,\ov{t_n})
            \tag{def.~$\freeinterTX{\cdot}$} \\
            &= \alpha \circ \freeinterX{\op} \circ (T\alpha)^n (\ov{t_1},\ldots,\ov{t_n})
            \tag{\cref{lem:homom_by_nat_of_mu}} \\
            &= \Parens{\op} \circ (\alpha \circ T\alpha)^n (\ov{t_1},\ldots,\ov{t_n}) 
            \tag{by \eqref{eqn:formula_a_opmx}} \\
            &= \Parens{\op} \circ (\alpha \circ \mu_X)^n (\ov{t_1},\ldots,\ov{t_n}) 
            \tag{by I.H.} \\
            &= \alpha \circ \freeinterX{\op} \circ \mu_X^n (\ov{t_1},\ldots,\ov{t_n}) 
            \tag{by \eqref{eqn:formula_a_opmx}} \\
            &= \alpha \circ \mu_X \circ \freeinterTX{\op} (\ov{t_1},\ldots,\ov{t_n})
            \tag{\cref{lem:homom_by_nat_of_mu}} \\
            &= \alpha \circ \mu_X \big( \ov{\op(t_1,\ldots,t_m)} \big)
            \tag{def.~$\freeinterTX{\cdot}$} \\
            &= \alpha \circ \mu_X \big( \ov{t} \big) 
        \end{align*}
    \end{itemize}
    This concludes the proof that $\alpha$ is associative.
\end{proof}

\begin{lemma}\label{lem:H_well_defined_functoriality}
   $H$ maps $(\Sigma\semifree,E\semifree)$-algebra homomorphisms to $T$-semialgebra homomorphisms.
\end{lemma}

\begin{proof}
    Suppose $f : (X, \Parens{\cdot}^X) \to (Y, \Parens{\cdot}^Y)$ is a $(\Sigma\semifree_{}, E\semifree_{})$-algebra homomorphism.
    Let $(X,\alpha) \defeq H (X, \Parens{\cdot}^X)$ and $(Y,\beta) \defeq H (Y, \Parens{\cdot}^Y)$.
    We want to check that $H(f)$, which is equal to $f$, is a $T$-semialgebra homomorphism, i.e. the following commute:
    \begin{center} \begin{tikzcd}
        TX
            \ar[d, "\alpha"']
            \ar[r, "Tf"]
        & TY
            \ar[d, "\beta"] \\
        X 
            \ar[r, "f"']
        & Y
    \end{tikzcd} \end{center}
    We prove $f \circ \alpha (\ov{t}) = \beta \circ Tf (\ov{t})$ for all $\ov{t} \in TX = \freealgebra{\Sigma}{X}{E}$ by an induction on $t$:
    \begin{itemize}
        \item 
        Suppose $t$ is some $x \in X$:
        \begin{align*}
            f \circ \alpha (\ov{x}) &= f \circ \alpha \circ \eta_X (x)
            \tag{def.~$\eta_X$} \\
            &= f \circ \Parens{\newa}^X (x) 
            \tag{by \eqref{eqn:formula_a_eta}} \\
            &= \Parens{\newa}^Y \circ f (x) 
            \tag{$f$ homom.} \\
            &= \beta \circ \eta_Y \circ f(x) 
            \tag{by \eqref{eqn:formula_a_eta}} \\
            &= \beta \circ Mf \circ \eta_X(x) 
            \tag{$\eta$ nat.} \\
            &= \beta \circ Mf (\ov{x}). 
            \tag{def.~$\eta_X$}
        \end{align*}
        
        \item 
        Suppose that it holds for $t_1,\ldots, t_n$ and let us prove it for $t = \op(t_1,\ldots,t_n)$:
        \begin{align*}
            f \circ \alpha (\ov{\op(t_1,\ldots,t_n)})
            &= f \circ \alpha \circ \freeinterX{\op} (\ov{t_1}, \ldots, \ov{t_n})
            \tag{def.~$\freeinterX{\cdot}$} \\
            &= f \circ \Parens{\op}^X (\alpha (\ov{t_1}), \ldots, \alpha (\ov{t_n}))
            \tag{by \eqref{eqn:formula_a_opmx}} \\
            &= \Parens{\op}^Y (f \circ \alpha)^n (\ov{t_1}, \ldots, \ov{t_n}) 
            \tag{$f$ homom.} \\
            &= \Parens{\op}^Y (\beta \circ Tf)^n (\ov{t_1}, \ldots, \ov{t_n}) 
            \tag{by I.H.} \\
            &= \beta \circ \freeinterY{\op} \circ (Tf)^n (\ov{t_1}, \ldots, \ov{t_n}) 
            \tag{by \eqref{eqn:formula_a_opmx}} \\
            &= \beta \circ Tf \circ \freeinterX{\op}(\ov{t_1}, \ldots, \ov{t_n}) 
            \tag{\cref{lem:homom_by_nat_of_mu}} \\
            &= \beta \circ Tf (\ov{\op(t_1,\ldots,t_m)}). 
            \tag{def.~$\freeinterX{\cdot}$}
        \end{align*}
    \end{itemize}
\end{proof}

\subsection{Joining both constructions} \label{sec:joining}

We have shown in the two previous sections that we have functors $G$ and $H$ as shown here:
    \[
        G: \EMs{T} \rightleftarrows \catalg{\Sigma_{}\semifree, E_{}\semifree} : H.
    \]
It remains to show that they are inverses.

\begin{lemma} \label{lem:G_H_inverses}
    The functors $G$ and $H$ are inverses.
\end{lemma}

\begin{proof}
    \begin{itemize}
        \item
        Suppose we start with a $T$-semialgebra $\alpha: TX \to X$.
        Let $(X, \Parens{\cdot}) \defeq G(X,\alpha)$ and $(X,\alpha') \defeq H (X, \Parens{\cdot})$.
        We prove $\alpha'(\ov{t}) = \alpha(\ov{t})$ for all $\ov{t} \in TX = \freealgebra{\Sigma}{X}{E}$ by induction on $t$:
        \begin{itemize}
            \item
            Suppose $t$ is some $x \in X$:
            \begin{align*}
                \alpha' (\ov{x}) &= \alpha' \circ \eta_X (x)
                \tag{def.~$\eta_X$} \\
                &= \Parens{\newa} (x) 
                \tag{by \eqref{eqn:formula_a_eta}} \\
                &= \alpha \circ \eta_X (x)
                \tag{def. of $\Parens{\newa}$ in \eqref{eqn:def_newa_interpretation}} \\
                &= \alpha \circ (\ov{x}).
                \tag{def.~$\eta_X$}
            \end{align*}
            
            \item
            Suppose it holds for $t_1, \ldots, t_n$ and let us prove it for $t = \op (t_1,\ldots,t_n)$:
            \begin{align*}
                \alpha' ( \ov{t} ) 
                &= \alpha' \big( \ov{\op(t_1,\ldots,t_n)} \big)
                \\
                &= \alpha' \circ \freeinterX{\op} (\ov{t_1}, \ldots, \ov{t_m})
                \tag{def.~$\freeinterX{\cdot}$} \\
                &= \Parens{\op} (\alpha' (\ov{t_1}) , \ldots , \alpha' (\ov{t_n})) 
                \tag{by \eqref{eqn:formula_a_opmx}} \\
                &= \Parens{\op} (\alpha (\ov{t_1}) , \ldots , \alpha (\ov{t_n}))
                \tag{by I.H.} \\
                &= \alpha \circ \freeinterX{\op} \circ \eta_X^n \circ \alpha^n (\ov{t_1} , \ldots , \ov{t_n}) 
                \tag{def.~$\Parens{\op}$ in \eqref{eqn:def_op_interpretation}} \\
                &= \alpha \circ \freeinterX{\op} (\ov{t_1} , \ldots , \ov{t_n})
                \tag{by \eqref{eqn:formula_a_opX}} \\
                &= \alpha \big( \ov{\op(t_1,\ldots,t_n)} \big)
                \tag{def.~$\freeinterX{\cdot}$} \\
                &= \alpha ( \ov{t} )
            \end{align*}
        \end{itemize}
        
        \item 
        Suppose we start with an $(\Sigma\semifree, E\semifree)$-algebra $(X,\Parens{\cdot})$.
        Let $(X,\alpha) \defeq H (X,\Parens{\cdot})$ and $(X,\Parens{\cdot}') \defeq G(X,\alpha)$.
        We want to prove that $\Parens{\cdot}' = \Parens{\cdot}$.
        Let us check first for $(\newa : 1) \in \Sigma\semifree$:
        \begin{align*}
            \Parens{\newa}' 
            &= \alpha \circ \eta_X 
            \tag{def.~$\Parens{\cdot}'$ in \eqref{eqn:def_newa_interpretation}} \\
            &= \Parens{\newa}, 
            \tag{by \eqref{eqn:formula_a_eta}}
        \end{align*}
        and then for $(\op : n) \in \Sigma_{}$:
        \begin{align*}
            \Parens{\op}' &= \alpha \circ \freeinterX{\op} \circ (\eta_X)^n
            \tag{def.~$\Parens{\cdot}'$ in \eqref{eqn:def_op_interpretation}} \\
            &= \Parens{\op} \circ (\alpha \circ \eta_X)^n 
            \tag{by \eqref{eqn:formula_a_opmx}} \\
            &= \Parens{\op} \circ \Parens{\newa}^n 
            \tag{by \eqref{eqn:formula_a_eta}} \\
            &= \Parens{\op}.
            \tag{equation \eqref{eqn:a_inside_disappear} in $E_{}\semifree$}
        \end{align*}
    \end{itemize}
\end{proof}

The proof of \cref{thm:big_theorem} is now complete and contained in Lemmas \ref{lem:G_well_defined_equations_satisfied}, 
\ref{lem:G_well_defined_functoriality},
\ref{lem:well_definedness_of_a},
\ref{lem:a_is_associative},
\ref{lem:H_well_defined_functoriality}, and
\ref{lem:G_H_inverses}.

\section{Examples}\label{sec:examples}

We now give multiple examples to illustrate the applicability of \cref{thm:conjecture}.
First, notice that some equations in $E\semifree$ can be simplified to equations in $E$. 
More precisely, equations in $E\semifree$ that arise via  \eqref{eqn:a_inside_terms_disappear} from an equation $t=s$ in $E$ where $t$ and $s$ are terms of depth at least 1, reduces to $t=s$ due to \cref{lem:generalisation_eq_23}.
Similarly, if $t$ is variable $v$ and $s$ is not a variable, then the equation in $E\semifree$ arising from \eqref{eqn:a_inside_terms_disappear} becomes $\newa v = s$; the case when $s$ is a variable and $t$ is not is analogous.

The presentations of the semifree monads on the maybe monad $(-)+1$, the semigroup monad $(-)^+$ and the distribution monad $\distribution$ were proven each individually by hand 
in \cite{Petrisan_Sarkis_2021}.
Those three examples gave a strong intuition that a uniform presentation was possible, which lead the authors of \cite{Petrisan_Sarkis_2021} to conjecture \cref{thm:conjecture} that we proved in this article.

\begin{example} \label{ex:exception}
    The \emph{exception} monad $X \mapsto X + K$, where $K$ is a fixed set (meant to contain a list of possible exception states) is presented by the theory of $K$-pointed sets $\Sigma = \setvbar{c_k : 0}{k \in K}$, $E = \emptyset$. 
    Its semifree monad has functor $X \mapsto X + (X + K)$, and presentation
    \begin{align*}
        \Sigma\semifree &= \setvbar{c_k : 0}{k \in K} \cup \set{\newa : 1}, \\
        E\semifree &= \set{\newa \newa v = \newa v} \cup \setvbar{\newa c_k = c_k}{k \in K}.
    \end{align*}
    Its algebras $(X,\Brackets{\cdot})$ are sets $X$ with a retract $Y = \ima (\Brackets{\newa})  \subseteq X$, the retraction being $\Brackets{\newa} : X \to X$,
    and with a set of distinguished elements $\setvbar{y_k}{k\in K} \subseteq Y$ in the retract.
\end{example}

\begin{example}\label{ex:list}
    The \emph{list} monad $X \mapsto X^* = \bigsqcup_{n \geq 0} X^n$ is presented by the theory of monoids $\Sigma_* = \set{e : 0,\cdot:2}$, $E_* = \set{(u \cdot v) \cdot w = u \cdot (v \cdot w),\ \  e \cdot v = v,\ \  v \cdot e = v}$, see e.g. \cite[Example 10.7]{Awodey_2006}.
    Its semifree monad has functor $X \mapsto X + X^*$.
    Notice that its presentation can be simplified.
    By \eqref{eqn:a_inside_terms_disappear} and \cref{lem:generalisation_eq_23}, we obtain the equations $e \cdot v = \newa v$ and $v \cdot e = \newa v$.
    These equations show that the new symbol $(\newa:1)$ is not needed since it can be expressed using the operations $(e:0)$ and $(\cdot:2)$.
    However, we then need to add the equation $e \cdot v = v \cdot e$.
    If we continue this simplification, we end up with
    \[
            E\semifree_* = \left\{
            \begin{aligned}
                e \cdot v &= v \cdot e,
                & e \cdot e &= e,\\
                e \cdot (u \cdot v) &= u \cdot v,
                & (u \cdot v) \cdot w &= u \cdot (v \cdot w)
            \end{aligned}
            \right\}.
    \]
    This corresponds to the theory of semigroups $(S,\cdot)$ that admit a retract $R = S \cdot S$ which is a monoid $(R,\cdot,e)$. The retraction is $e \cdot (-) = (-) \cdot e: S \to S$.
\end{example}

\begin{example}\label{ex:multiset}
    The \emph{multiset} monad $\multiset(X) = \setvbar{\phi : X \to \N}{\supp{\phi} \text{ finite}}$, where $\supp{\phi}$ means the support of $\phi$, is presented  by the theory of commutative monoids $\Sigma_\multiset = \set{e : 0,\cdot:2}$, $E_\multiset = E_* \uplus \set{u \cdot v = v \cdot u}$ \cite[Section 2]{Jacobs_2010}.
    Its semifree monad has functor $\multiset\semifree(X) = X + \multiset(X)$.
    Its presentation can be built direct on the simplified presentation of the semifree monoid monad of \cref{ex:list}.
    Furthermore, $u \cdot v = v \cdot u$ renders the equation $e \cdot v = v \cdot e$ redundant.
    Therefore, this presentation ends up being $\Sigma\semifree_\multiset = \set{e:0, \cdot:2}$ and
    \[
            E\semifree_* = \left\{
            \begin{aligned}
                u \cdot v &= v \cdot u,
                & e \cdot e &= e,\\
                e \cdot (u \cdot v) &= u \cdot v,
                & (u \cdot v) \cdot w &= u \cdot (v \cdot w)
            \end{aligned}
            \right\}.
    \]
    This corresponds to the theory of commutative semigroups $(S,\cdot)$ that admits a retract $R = S \cdot S$, which is a commutative monoid $(R,\cdot,e)$. The retraction is $e \cdot (-): S \to S$.
\end{example}

\begin{example}\label{ex:powerset}
    The \emph{finite powerset} monad $\powerset(X) = \setvbar{Y \subseteq X}{Y \text{ finite}}$ is presented   by the theory of join-semilattices with bottom $\Sigma_\powerset = \set{e : 0,\cdot:2}$, $E_\powerset = E_\multiset \uplus \set{v \cdot v = v}$ \cite[p.~81]{Jacobs_1994}.
    Its semifree monad has functor $\powerset\semifree(x) = X + \powerset(X)$.
    The equation $v \cdot v = v$ becomes $v \cdot v = \newa v$, and as in \cref{ex:list}, the symbol $\newa$ can be replaced by $e \cdot -$.
    We end up with
    \[
        \Sigma\semifree_\powerset = \set{e:0, \cdot:2},
        \qquad
        E\semifree_\powerset = E\semifree_\multiset \uplus \set{v \cdot v = e \cdot v}.
    \]
\end{example}

\begin{example}\label{ex:state}
        The \emph{state} monad $\state(X) = (S \times X)^S$, where $S$ is a fixed finite set (generally of states), is presented  by the theory of global states (see \cite{Metayer_2004,Plotkin_Power_2001}) $\Sigma_\state = \set{f:n} \cup \setvbar{g_i:1}{1 \leq i \leq n}$, where $n \defeq \Card{S}$, and
    \[
        E_\state =
        \left\{
        \begin{aligned}
            g_i g_j v &= g_j v, & 
            g_i f (v_1, \ldots, v_n) &= g_i v_i, 
            & 
            f(g_1 v, \ldots, g_n v) = v
        \end{aligned}
        \right\}.
    \]
    The semifree monad on $\state$ has functor $\state\semifree(x) = X + (S \times X)^S$.
    Its presentation can be simplified.
    The equation $f(g_1 v, \ldots, g_n v) = \newa v$ implies that $\newa$ can be expressed as $f(g_1 (-), \ldots, g_n (-))$.
    Then, some equations turn out to be redundant, like $g_i \newa v = g_i v$ or $\newa \newa v = \newa v$.
    After simplifications, we end up with $\Sigma\semifree_\state = \set{f:n} \cup \setvbar{g_i:1}{0 \leq i \leq n}$, and
    \[
        E\semifree_\state =
        \left\{
            \begin{aligned}
                f(g_1 v_1, \ldots, g_n v_n) &= f(v_1, \ldots, v_n), & g_i g_j v &= g_j v,\\
                f(g_i v, \ldots, g_i v) &= g_i v, & g_i f (v_1, \ldots, v_n) &= g_i v_i
            \end{aligned}
        \right\}.
    \]
\end{example}

\begin{example}\label{ex:id_sn}
    Consider the repeated semifree construction on the identity monad $\Id$ on $\catset$.
    \begin{itemize}
        \item
        $\Id$ is presented by $\Sigma = E = \emptyset$.
        Note that $\Id$-algebras are identity morphisms because of the unit axiom \eqref{eqn:unit_axioms}, 
        and $(\Sigma,E)$-algebras are just sets with no operations.
        The presentation sends an $\Id$-algebra $(X,\id_X)$ to the $(\Sigma,E)$-algebra $(X,\emptyset)$.
        
        \item
        $\Id\semifree(X) = X + X$ 
        is presented by $\Sigma\semifree = \set{\newa : 1}$ and $E\semifree = \set{\newa \newa v = \newa v}$.
        
        \item
        $\Id\nsemifree{2}(X) = X + (X+X)$ is presented by $\Sigma\nsemifree{2} = \set{\newa, \newb : 1}$ and
        \[
            E\nsemifree{2} =
            \set{
                \newb \newb v = \newb v, \quad
                \newb \newa v = \newa v, \quad
                \newa \newb v = \newa v, \quad
                \newa \newa v = \newa v
            }.
        \]
        The equation directly given by \eqref{eqn:a_inside_terms_disappear} is $\newa \newa \newb v = \newa \newb v$, which simplifies to $\newa \newa v = \newa v$ using \cref{lem:generalisation_eq_23}.
        We can summarize $E\nsemifree{2}$ as saying that $\newa$ and $\newb$ are idempotent, and that $\newa$ absorbs $\newb$ on the left and on the right.
        
        \item
        Repeating the procedure $n$ times, we inductively obtain a monad on $n+1$ disjoint copies of $X$, $\Id\nsemifree{n}(X) = X + \ldots + X \isom n_\infty \times X$, where $n_\infty \defeq n \uplus \set{\infty} = \set{0, \ldots, n-1, \infty}$.
        The set $n_\infty$ is a linear order, as a subset of the extended natural numbers $\N \uplus \set{\infty}$.
        It is also a monoid, with unit $\infty$ and operation $\mathsf{min}$.
        Hence, we have a writer monad. 
        Its presentation is given by $n$ unary idempotents $\Sigma\nsemifree{n} = \set{\newa_0,\ldots,\newa_{n-1} : 1}$.
        An idempotent $\newa_i$ absorbs another one $\newa_j$ on both sides whenever $i < j$:
        \[
            E\nsemifree{n} =
            \setvbar{\newa_i \newa_j v = \newa_{\mathsf{min}(i,j)} v}{0 \leq i,j \leq n-1}.
        \]
        An $\id\nsemifree{n}$-algebra with carrier set $X$ corresponds then to having a left-action of 
        $n_\infty$ on $X$.
        This is not a surprise, as algebras of writer monads are always actions.
        Being a linear order, $n_\infty$ is also a meet-semilattice.
        Such an $\id\nsemifree{n}$-algebra can thus be viewed as a semi-lattice automaton, similar to the lattice automata of \cite{Kupferman_Lustig_2007}. The carrier $X$ is the set of states, the $\newa_i$ are the letters of the alphabet, and a transition on $\newa_i$ is simply associated with the value of $\newa_i$ in the semilattice $n_\infty$.  
    \end{itemize}
\end{example}

\section{Relation between semifree monads and other monad constructions} \label{sec:ideal_monad}

\emph{Ideal monads} were introduced by Aczel et al.~\cite{Aczel_Adamek_Milius_Velebil_2003} to study solutions to guarded recursive equations using coalgebraic methods.
The authors investigated completely iterative monads, following earlier work from Elgot et al.~\cite{Elgot_1978} on iterative algebraic theories.
Ideal monads are an abstraction of the core properties of completely iterative monads, and, in particular, completely iterative monads are ideal. 
Moreover, Ghani \& Uustalu~\cite{Ghani_Uustalu_2004} showed that ideal monads have the right mathematical structure to avoid the problems encountered in the construction of monad coproducts, and they gave a simple construction of coproducts of ideal monads by investigating coalgebraic fixed points.
This same construction was later proved to be applicable to the class of consistent $\catset$-monads \cite{Adamek_Milius_Bowler_Levy_2012}.

We now define ideal monads and show that semifree monads are ideal.

\begin{definition} \label{def:ideal_monad}
    In a category with finite coproducts, an \myemph{ideal monad} is a quintuple $(T,\eta,\mu,T_0,m_0)$ where $(T,\eta,\mu)$ is a monad with functor $T = \Id + T_0$, unit $\eta = \inl^{\Id+T} : \Id \to \Id + T_0$, and where $\mu : TT \to T$ ``restricts'' to the natural transformation $m_0 : T_0T \to T_0$, meaning that the following square commutes.
    \begin{equation} \label{eqn:diagram_def_ideal_monad}
        \begin{tikzcd}[column sep=large]
            T_0T
                \ar[r, "\inr^{\Id + T} \circ T"]
                \ar[d, "m_0"']
            & TT
                \ar[d, "\mu"] \\
            T_0
                \ar[r, "\inr^{\Id + T}"']
            & T
        \end{tikzcd}
    \end{equation}
\end{definition}

\noindent Notice that having \eqref{eqn:diagram_def_ideal_monad} commute is equivalent to having
\begin{equation} \label{eqn:def_ideal_monad}
    \mu = [\Id_{\Id + T_0}, \inr^{\Id + T_0} \circ m_0].
\end{equation}

Examples of ideal monads include 
(free) completely iterative monads \cite[Example~4.4]{Aczel_Adamek_Milius_Velebil_2003},
algebraically free monads,
exception monads, and interactive output monads
\cite{Ghani_Uustalu_2004}.



\begin{example}
    Semifree monads are ideal.
    Take a semifree monad $(M\semifree,\eta\semifree,\mu\semifree)$.
    It is enough to show that its multiplication $\mu\semifree$ satisfies \eqref{eqn:def_ideal_monad}.
    For that, let
    \[
        m_0 \defeq \big( M M\semifree \xrightarrow{\smash{ M[\eta,\Id_M] } } MM \xrightarrow{\mu} M \big)
    \]
    and notice that the equation on the right-hand side of \eqref{eqn:def_ideal_monad} becomes the exact definition of $\mu\semifree$ in \cref{def:semifree_monad}.
    Hence, $(M\semifree,\eta\semifree,\mu\semifree,M,m_0)$ is ideal.
\end{example}

Semifree monads thus enjoy all the nice properties of ideal monads such as having simple coproducts \cite{Ghani_Uustalu_2004}.

The concrete isomorphism between $M\semifree$-algebras and $M$-semialgebras can be phrased as, any semialgebra structure for $M$ on an object $X$ can be replaced by an Eilenberg-Moore algebra structure for $M\semifree$. 
We will show that this property has an analogue for ideal monads. 
However, since $T_0$ is only assumed to be an endofunctor, there is no notion of $T_0$-semialgebra, but we can use the ideal monad structure to define an analogue of the associativity diagram of Eilenberg-Moore algebras  \eqref{eqn:associativity_axiom2}. 
This leads us to consider functor $T_0$-algebras (i.e., morphisms of the type $a \colon T_0 X \to X$), such that a square, analogous to the associativity square of \eqref{eqn:associativity_axiom2}, commutes (diagram \eqref{eqn:analogue_semialgebra_algebra_for_ideal_monads} below).
We denote the category of functor $T_0$-algebras and $T_0$-algebra morphisms by $\catalg{T_0}$.

\begin{lemma} \label{lem:analogue_semialgebra_algebra_for_ideal_monads}
    Let $(T,\eta,\mu,T_0,m_0)$ be an ideal monad.
    There is an isomorphism between $\EM{T}$ and the full subcategory of $\catalg{T_0}$ of all $a:T_0 X \to X$ that makes \eqref{eqn:analogue_semialgebra_algebra_for_ideal_monads} commute.
    \begin{equation} \label{eqn:analogue_semialgebra_algebra_for_ideal_monads}
        \begin{tikzcd}[column sep=large]
            T_0 T X
                \ar[r, "{T_0 [\Id_X, a]}"]
                \ar[d, "m_{0,X}"']
            & T_0 X
                \ar[d, "a"] \\
            T_0 X \ar[r, "a"']
            & X
        \end{tikzcd}
    \end{equation}
\end{lemma}

As observed in \cref{lem:semifree_construction_is_a_functor}, the semifree construction is an endofunctor on the category of monads $\catmon{C}$.
It is therefore a natural question to ask whether $(-)\semifree$ is a \emph{monad transformer}, i.e., a pointed endofunctor on $\catmon{C}$ \cite{Liang_Hudak_Jones_1995}. 
Here, \emph{pointed} means to admit a natural transformation $\id_\catmon{C} \Rightarrow (-)\semifree$.
It turns out that this functor is not pointed, but it is \emph{co-pointed}, meaning that it admits a natural transformation $\epsilon\colon (-)\semifree \Rightarrow \id_\catmon{C}$. We have, in fact, a comonad $((-)\semifree,\epsilon,\delta)$ on $\catmon{C}$.

We collect the above observations in the following lemmas.

\begin{lemma} \label{lem:semifree_construction_not_pointed_endofunctor}
    Given a category $\cat{C}$ with finite coproducts, the semifree endofunctor $(-)\semifree : \catmon{C} \to \catmon{C}$ is not pointed.
\end{lemma}

\begin{lemma}\label{lem:copointed_endofunctor}
    Given a category $\cat{C}$ with finite coproducts, the triple $((-)\semifree, \epsilon, \delta)$ is a comonad on $\catmon{C}$ by
     defining for a monad $(M,\eta,\mu)$ and an object $X$:
\begin{align*}
        \epsilon_{M,X} : X+MX &\xrightarrow{[\eta_X,\id_{MX}]} MX, \text{ and} \\
        \delta_{M,X} : X + MX &\xrightarrow{\id_X + \inr^{X+MX}} X + (X + MX).
    \end{align*}
\end{lemma}

As consequence of the above, the semifree construction is not a monad transformer, at least not with the straightforward definitions.
In general, a monad transformer $T_L$ is defined on a ``base'' monad $L$ such that  $L = T_L(\id_\cat{C})$.
In our case, we have $(\id_\cat{C})\semifree = \id_\cat{C} + \id_\cat{C}$. 
This is an interesting monad in itself, even though it was not the starting point of the semifree construction.

\newpage

\section{Conclusion}\label{sec:future_work}
In this paper, we proved a uniform algebraic presentation of semifree monads $M\semifree(X) = X + M(X)$ when an algebraic presentation of the monad $M$ is known.
We also showed that semifree monads are instances of ideal monads, and that the semifree construction is not a monad transformer, but it is a comonad on the category of monads. 


There are several directions for future work.
Given that the functor part of the semifree monad $M\semifree$ is a functor coproduct $\Id + M$, it would be interesting to understand better the relationship to coproducts of monads, and whether \cref{thm:conjecture} could be generalised to give presentations of (certain) monad coproducts~\cite{Ghani_Uustalu_2004}. 
Similarly, we would like to investigate if the observations we made in \cref{ex:id_sn} on algebras for the iterated semifree monad on identity can be generalised to other monads.


The presentation of semifree monads provides a means to study no-go theorems for weak distributive laws, using the correspondence between weak distributive laws for $M$ and certain distributive laws for $M\semifree$ \cite{Petrisan_Sarkis_2021}, and a similar approach as in \cite{Zwart_Marsden_2019}.
In \cite{Zwart_Marsden_2019}, no-go theorems for distributive laws are proved using criteria on presentations of the monads.
However, their no-go theorems are not directly applicable for the semifree monad, as they require equations where one side is a single variable, e.g., unitality $e * v = v$ or idempotence $v * v = v$.
Such equations never occur in the presentations of $M\semifree$, as the simplest terms that those equations contain are of the form $\newa v$. 
Hence new problematic equations in presentations of monads must be identified.
Furthermore, weak distributive laws for $M$ correspond to distributive laws for $M\semifree$ satisfying an extra condition, and hence to a subclass of the composite theories for $M\semifree$.
This subclass may satisfy more equations, which could be helpful in establishing no-go theorems for weak distributive laws for $M$.


A more fundamental question related to the definition of algebraic presentation is whether, over the category $\catset$, an isomorphism between Eilenberg-Moore categories of algebras implies an isomorphism of $\catset$-monads.
This fails over general categories, but an argument may exist for the specific case of $\catset$. 

Finally, since $\catset$-monads sometimes have interesting liftings to other categories, one could consider the following question.
Suppose we have a uniform presentation for a construction on a certain $\catset$-monad which can be lifted to a category $\cat{C}$, 
does this presentation also lift to $\cat{C}$?
This was investigated in \cite{Mio_Sarkis_Vignudelli_2021} for the non-empty convex distribution monad and its Hausdorff–Kantorovich lifting to metric spaces.


\newpage
\bibliographystyle{meta/splncs04}
\bibliography{main}
\newpage

\section{Appendix}

\begin{lemma} \label{lem:free_object_preserve_by_concrete_isom}
    Concrete isomorphisms preserve free objects
\end{lemma}

\begin{proof}[Proof of \cref{lem:free_object_preserve_by_concrete_isom}]
    Take concrete categories $U : \cat{C} \to \catset$ and $U' : \cat{C'} \to \catset$, and a concrete isomorphism $F : C \rightleftarrows C' : F\inv.$
    Let $X$ be a set and $Y \in \cat{C}$ be free on $X$ (w.r.t.~$U$).
    We want to prove that $FY$ is also free on $X$ (w.r.t.~$U'$).
    Since $Y$ is free on $X$, there exists an arrow $\eta_X : X \to UY$.
    We want an arrow $X \to U'FY$, but $F$ being concrete implies $U'F(Y) = UY$, thus we can consider the same arrow $\eta_X$.
    To prove the universal property, take $Z \in C'$ with an arrow $f : X \to U'Z$.
    Consider $F\inv Z$ in $\cat{C}$.
    It comes with the same arrow ${f: X \to U'Z  = UF\inv(Z)}$.
    By the universal property of $Y$ being free on $X$, there exists a unique $\tilde{f} : Y \to F\inv Z$ such that $f = U\tilde{f} \circ \eta_X$.
    This is equivalent to saying that there exists a unique function $F\tilde{f} : FY \to Z$ such that $U'F\tilde{f}$ such that $f = U'F\tilde{f} \circ \eta_X$, concluding the proof.
    \begin{center}
        \begin{tikzpicture}[mybox/.style={draw, inner sep=5pt, rounded corners}]
            \node[mybox,
                    label=west:{$\cat{C}$}
                    ] at (-2,0) {%
                \begin{tikzcd}
                    Y \ar[r, "\tilde{f}"] & F\inv Z
                \end{tikzcd}
            };
            \node[mybox,
                    label=east:{$\cat{C'}$}
                    ] at (2,0) {%
                \begin{tikzcd}
                    FY \ar[r, "F\tilde{f}"] & Z
                \end{tikzcd}
            };
            \node[mybox,
                    label=west:{$\catset$}
                    ] at (0,-2) {%
                \begin{tikzcd}
                    UY = U'FY \ar[rr, "U\tilde{f} = U'F\tilde{f}"] & & U'Z = UF\inv Z \\
                    & X \ar[ul, "\eta_X"] \ar[ur, "f"'] & 
                \end{tikzcd}
            };
        \end{tikzpicture}
    \end{center}
    
\end{proof}

\begin{proof}[Proof of \cref{lem:isom_monads_iff_isom_algebras_cats}]
    $(\Rightarrow)$ Suppose we have a monad isomorphism $\sigma : M \xrightarrow{\isom} T$.
    We construct two concrete functors:
    \[{
        \begin{aligned}
            F: \EM{M} &\rightleftarrows \EM{T} : F\inv \\
            (MX \xrightarrow{\alpha} X) &\mapsto (TX \xrightarrow{\sigma_X\inv} MX \xrightarrow{\alpha} X) \\
            (MX \xrightarrow{\sigma_X} TX \xrightarrow{\beta} X) &\mapsfrom (TX \xrightarrow{\beta} X)
        \end{aligned}
    }\]
    Let us verify that the image of $F$ are indeed $T$-algebras:
    \begin{center}
        \begin{tikzcd}
            X
                \ar[r, "\eta^T_X"]
                \ar[dr, "\eta^M_X" description, bend right]
                \ar[ddr, equal, bend right]
                \ar[dr, phantom, "\scriptstyle \eqref{eqn:monad_map_axiom1}", shift left = .3em]
            & TX \ar[d, "\sigma_X\inv"] \\
            {} 
                \ar[dr, phantom, "\scriptstyle \eqref{eqn:unit_axioms2}", shift left = .7em] 
            & MX \ar[d, "\alpha"] \\
            & X
        \end{tikzcd}
        \qquad \qquad
        \begin{tikzcd}
            T^2 X
                \ar[rr, "\mu^T_X"]
                \ar[d, "T \sigma_X\inv"']
                \ar[drr, phantom, "\scriptstyle \eqref{eqn:monad_map_axiom2}"]
            &
            & TX
                \ar[d, "\sigma_X\inv"] \\
            TMX
                \ar[r, "\sigma_{MX}\inv"]
                \ar[d, "T\alpha"']
                \ar[dr, "\scriptstyle (\sigma\inv \text{ nat.})", phantom]
            & MMX
                \ar[r, "\mu_X"]
                \ar[d, "M\alpha" description]
                \ar[dr, phantom, "\scriptstyle \eqref{eqn:associativity_axiom2}"]
            & MX
                \ar[d, "\alpha"] \\
            TX
                \ar[r, "\sigma_X\inv"']
            & MX
                \ar[r, "\alpha"']
            & X
        \end{tikzcd}
    \end{center}
    Similarly, $F\inv$ is also a valid functor.
    The functors $F$ and $F\inv$ are concrete and are clearly inverse of each other.
    
    $(\Leftarrow)$ Suppose we have a concrete isomorphism 
    \[
        F: \EM{M} \rightleftarrows \EM{T} : F\inv.
    \]
    Take a set $X$.
    In $\EM{M}$, $(MX, \mu^M_X)$ is a free object on $X$.
    By \cref{lem:free_object_preserve_by_concrete_isom}, the object $(TX,\alpha) \defeq F\inv(TX,\mu^T_X)$ is also free on $X$.
    By uniqueness of free objects, there exists an isomorphism $\sigma_X$ between the two and it must, in particular, commute with the units:
    \begin{equation} \label{eqn:isomomorphism_commute_with_units} 
        \begin{tikzcd} 
            (MX,\mu^M_X)
                \ar[rr, "\sigma_X", shift left]
            &
            & (TX, \alpha)
                \ar[ll, "\sigma_X\inv", shift left] \\
            & X 
                \ar[ul, "\eta^M_X"]
                \ar[ur, "\eta^T_X"']
            &
        \end{tikzcd}
    \end{equation}
    We want to prove that $\sigma$ is natural in $X$.
    Take $f:X \to Y$ and let $(TY,\beta) \defeq F\inv(TY,\mu^T_Y)$.
    We will show that $\sigma_Y \circ Mf = Tf \circ \sigma_X$ by using the universal property of $(MX,\mu^M_X)$ as a free objects. More precisely, by showing that $\sigma_Y \circ Mf $ and $Tf \circ \sigma_X$ are both extensions of $\eta^T_Y \circ f$, and hence they must be equal. This amounts to showing that the triangle diagram below in $\catset$ commutes:
    \begin{center}
        \begin{tikzpicture}[mybox/.style={draw, inner sep=5pt, rounded corners}]
            \node[mybox,
                    label=west:{$\EM{M}$}
                    ] at (0,0) {%
                \begin{tikzcd}
                    (MX,\mu^M_X) 
                        \ar[rr, shift left, "\sigma_Y \circ Mf"]
                        \ar[rr, shift right, "Tf \circ \sigma_X"']
                    &
                    & (TY, \beta)
                \end{tikzcd}
            };
            \node[mybox,
                    label=west:{$\catset$}
                    ] at (0,-2.2) {%
                \begin{tikzcd}[row sep=small]
                    MX
                        \ar[rr, shift left, "\sigma_Y \circ Mf"]
                        \ar[rr, shift right, "Tf \circ \sigma_X"']
                    &
                    & TY \\
                    & Y
                        \ar[ur, "\eta^T_Y"']
                    & \\
                    X
                        \ar[ur, "f"']
                        \ar[uu, "\eta^M_X"]
                    & &
                \end{tikzcd}
            };
        \end{tikzpicture}
    \end{center}
    Verification:
    \begin{center}
        \begin{tikzcd}[sep = small]
            {}
                \ar[drr, phantom, "\scriptstyle \eqref{eqn:isomomorphism_commute_with_units}", pos =.55]
            & MX
                \ar[dr, "\sigma_X"]
            & \\
            X
                \ar[ur, "\eta^M_X"]
                \ar[rr, "\eta^T_X" description]
                \ar[d, "f"']
                \ar[drr, phantom, "\scriptstyle (\eta^T \text{ nat.})"]
            &
            & TX
                \ar[d, "Tf"] \\
            Y
                \ar[rr, "\eta^T_Y"']
            &
            & TY
        \end{tikzcd}
        \qquad \qquad
        \begin{tikzcd}[sep=small]
            X
                \ar[rr, "\eta^M_X"]
                \ar[d, "f"']
                \ar[drr, phantom, "\scriptstyle (\eta^M \text{ nat.})"]
            &
            & MX
                \ar[d, "Mf"] \\
            Y
                \ar[rr, "\eta^M_Y" description]
                \ar[dr, "\eta^T_Y"']
                \ar[drr, phantom,  "\scriptstyle \eqref{eqn:isomomorphism_commute_with_units}", pos =.45]
            &
            & MY
                \ar[dl, "\sigma_Y"] \\
            & TY
            & {}
        \end{tikzcd}
    \end{center}
    Similarly, $\sigma\inv$ is also natural in $X$.
    We thus have a natural isomorphism $\sigma : M \Rightarrow T$.
    It remains to check that it is a monad morphism.
    The first requirement $\sigma_X \circ \eta^M_X = \eta^T_X$ holds by definition of $\sigma$, as depicted in \eqref{eqn:isomomorphism_commute_with_units}.
    The second requirement requires us to prove that
    \begin{center}
        \begin{tikzcd}
            MMX
                \ar[r, "M \sigma_X"]
                \ar[d, "\mu^M_X"']
            & MTX
                \ar[r, "\sigma_{TX}"]
            & TTX
                \ar[d, "\mu^T_X"] \\
            MX
                \ar[rr, "\sigma_X"']
            &
            & TX
        \end{tikzcd}
    \end{center}
    commutes.
    Let us prove it using the universal property of free objects again:
    \begin{center}
        \begin{tikzpicture}[mybox/.style={draw, inner sep=5pt, rounded corners}]
            \node[mybox,
                    label=west:{$\EM{M}$}
                    ] at (0,0) {%
                \begin{tikzcd}
                    (MMX,\mu^M_{MX}) 
                        \ar[rr, shift left, "\mu^T_X \cdot \sigma_{TX} \cdot M\sigma_X"]
                        \ar[rr, shift right, "\sigma_X \cdot \mu^M_X"']
                    &
                    & (TY, \beta)
                \end{tikzcd}
            };
            \node[mybox,
                    label=west:{$\catset$}
                    ] at (0,-2.2) {%
                \begin{tikzcd}[row sep=small]
                    MMX
                        \ar[rr, shift left, "\mu^T_X \cdot \sigma_{TX} \cdot M\sigma_X"]
                        \ar[rr, shift right, "\sigma_X \cdot \mu^M_X"']
                    &
                    & TY \\
                    & {}
                    & \\
                    X
                        \ar[uurr, "\sigma_X"']
                        \ar[uu, "\eta^M_{MX}"]
                    & &
                \end{tikzcd}
            };
        \end{tikzpicture}
    \end{center}
    Verification:
    \begin{center}
        \begin{tikzcd}
            MX
                \ar[r, "\eta^M_{MX}"]
                \ar[d, "\sigma_X"']
                \ar[dr, phantom, "\scriptstyle (\eta^M \text{ nat.})"]
            & MMX
                \ar[d, "M \sigma_X"] \\
            TX
                \ar[r, "\eta^M_{TX}" description]
                \ar[dr, "\eta^T_{TX}" description, bend right]
                \ar[ddr, equal, bend right]
                \ar[dr, phantom, "\scriptstyle \eqref{eqn:isomomorphism_commute_with_units}", pos=.6, shift left = .2em]
            & MTX
                \ar[d, "\sigma_{TX}"] \\
            {}
                \ar[dr, phantom, "\scriptstyle \eqref{eqn:unit_axioms}", pos =.75, shift left = .5em]
            & TTX
                \ar[d, "\mu^T_X"] \\
            {}
            & TX
        \end{tikzcd}
        \qquad \qquad
        \begin{tikzcd}[sep=small]
            MX
                \ar[r, "\eta^M_{MX}"]
                \ar[dr, equal, bend right]
                \ar[dr, phantom, "\scriptstyle \eqref{eqn:unit_axioms}", pos =.75, shift left = .3em]
            & MMX
                \ar[d, "\mu^M_X"] \\
            {}
            & MX
                \ar[d, "\sigma_X"] \\
            {}
            & TX
        \end{tikzcd}
    \end{center}
    Similarly, $\sigma\inv$ is a monad morphism.
    Hence we have a monad isomorphism $\sigma : M \isom T : \sigma \inv$ concluding the proof.
\end{proof}

\begin{proof}[Proof of \cref{lem:homom_by_nat_of_mu}]
    Take a function symbol $(\op:n) \in \Sigma$.
    We have:
    \begin{center}
        \begin{tikzcd}
             (TTX)^n
                \ar[r, "\ov{\op}"]
                \ar[d, "\mu_X^n"']
                \arrow[urrd, to path= { 
                    -- ([yshift=1.5ex]\tikztostart.north) 
                    |- ([yshift=1.5ex]\tikztotarget.north) node[near end,above,font=\scriptsize]{${\freeinterTX{\op}}$}
                    -- (\tikztotarget)}]
                \ar[dr, phantom, "\scriptstyle (*)"] 
            & TTTX
                \ar[d, "T\mu_X" description]
                \ar[r, "\mu_{TX}"]
                \ar[dr, phantom, "\scriptstyle (\mu \text{ assoc.})"]
            & TTX \vphantom{(TTX)^n}
                \ar[d, "\mu_X"] \\
            (TX)^n
                \ar[r, "\ov{\op}"']
                \arrow[drru, to path= { 
                    -- ([yshift=-1.5ex]\tikztostart.south) 
                    |- ([yshift=-1.5ex]\tikztotarget.south) node[near end, below, font=\scriptsize] {${\freeinterX{\op}}$}
                    -- (\tikztotarget)}]
            & TTX 
                \ar[r, "\mu_X"']
            & TX \vphantom{(TX)^n}
        \end{tikzcd}
        \qquad
        \begin{tikzcd}
            (TY)^n
                \ar[d, "(Tf)^n"']
                \ar[r, "\ov{\op}"]
                \arrow[urrd, to path= { 
                    -- ([yshift=2ex]\tikztostart.north) 
                    |- ([yshift=2ex]\tikztotarget.north) node[near end,above,font=\scriptsize]{$\freeinterY{\op}$}
                    -- (\tikztotarget)}]
                \ar[dr, phantom, "\scriptstyle (*)"]
            & TTY
                \ar[d, "TTf" description]
                \ar[r, "\mu_Y"]
                \ar[dr, phantom, "\scriptstyle (\mu \text{ nat.})"]
            & TY \vphantom{(TY)^n}
                \ar[d, "Tf"] \\
            (TX)^n 
                \ar[r, "\ov{\op}"']
                \arrow[drru, to path= { 
                    -- ([yshift=-2ex]\tikztostart.south) 
                    |- ([yshift=-2ex]\tikztotarget.south) node[near end, below, font=\scriptsize] {$\freeinterX{\op}$}
                    -- (\tikztotarget)}]
            & TTX
                \ar[r, "\mu_X"'] 
            & TX \vphantom{(TX)^n}
        \end{tikzcd}
    \end{center}
    Both squares marked $(*)$ commute as both of their paths send
    \[
        \left( \ov{t_1}[\frac{\ov{s_{i,1}}}{v_{i,1}}], \ldots, \ov{t_n}[\frac{\ov{s_{i,n}}}{v_{i,n}}] \right)
        \quad \mapsto \quad
        \ov{\op} \left(  \ov{t_1[\frac{s_{i,1}}{v_{i,1}}]} , \ldots , \ov{t_n[\frac{s_{i,n}}{v_{i,n}}]}  \right).
    \]
\end{proof}

\begin{proof}[Proof of \cref{lem:semifree_construction_is_a_functor}]
    The verification of the functoriality of $(-)\semifree$ is clear from the definition \eqref{eq:semifree-functor-action-bis}.
    What really needs to be proven is that given a monad morphism $\sigma : M \Rightarrow T$, then $\sigma\semifree : M\semifree \Rightarrow T\semifree$ is also a monad morphism.
    The naturality of $\sigma$ implies the naturality of $\sigma\semifree$: for $f: X \to Y$ we have
    \[
        \begin{tikzcd}
            MX
                \ar[r, "\sigma_X"]
                \ar[d, "Mf"']
            & TX
                \ar[d, "Nf"] \\
            MY
                \ar[r, "\sigma_Y"']
            & TY
        \end{tikzcd}
        \qquad \Leftrightarrow \qquad
        \begin{tikzcd}
            X+MX
                \ar[r, "\id_X+\sigma_X"]
                \ar[d, "f+Mf"']
            & X+TX
                \ar[d, "f+Nf"] \\
            Y+MY
                \ar[r, "\id_Y+\sigma_Y"']
            & Y+TY
        \end{tikzcd}
    \]
    The first axiom for $\sigma\semifree$ to be a monad morphism is a simple application of the coproduct formula $(f+g) \circ \inl^{X+Y} = \inl^{X'+Y'} \circ f$, for $f:X \to X'$ and $g : Y \to Y'$:
    \[
        \begin{tikzcd}[row sep=small]
            & X + MX
                \ar[dd, "\id_X + \sigma_X"] \\
            X
                \ar[ur, "\inl^{X+MX}"]
                \ar[dr, "\inl^{X+NX}"']
            & \\
            & X + NX
        \end{tikzcd}
    \]
    The second axiom simply requires to carefully look what happens at each component:
    \[
        \begin{tikzcd}[every label/.append style={scale=.75}, cells={nodes={scale=.75}}]
            (X+MX) + M(X+MX)
                \ar[r, "M\semifree \sigma_X\semifree = (\id_X + \sigma_X) + M(\id_X + \sigma_X)"{yshift=4pt}]
                \ar[d, "{   [\id_{X+MX}   ,  \inr^{X+MX} \circ \mu^M_X \circ M[\eta^M_X,\id_{MX}]   ]   }" description]
            & (X + TX) + M(X + TX)
                \ar[r, "\sigma_{T\semifree X}\semifree = \id_{X+TX} + \sigma_{X+TX}"{yshift=4pt}]
            & (X + TX) + T(X + TX)
                \ar[d, "{   [\id_{X+TX}   ,  \inr^{X+TX} \circ \mu^T_X \circ T[\eta^T_X,\id_{TX}]   ]   }" description] \\
            X + MX
                \ar[rr, "\id_X + \sigma_X"']
            &
            & X + TX
        \end{tikzcd}
    \]
    The elements on the first, respectively the second component, go through $\id_X$, respectively $\sigma_X$, with both the bottom and the top path.
    On the third component, we have
    \[
        \begin{tikzcd}[every label/.append style={scale=.95}, cells={nodes={scale=.85}}]
            M(X+MX)
                \ar[r, "M(\id_X + \sigma_X)"{yshift=4pt}]
                \ar[d, "{  M[\eta^M_X, \id_{MX}]  }"']
                \ar[drr, phantom, "(*)"{yshift=4pt}]
            & M(X+TX)
                \ar[r, "\sigma_{X+TX}"{yshift=4pt}]
            & T(X+TX)
                \ar[d, "{  T[\eta^T_X, \id_{TX}]  }"] \\
            MMX
                \ar[r, "M\sigma_X" description]
                \ar[d, "\mu^M_X"']
                \ar[drr, phantom, "\small (\sigma \text{ monad morphism, \eqref{eqn:monad_map_axiom2}})"]
            & MTX
                \ar[r, "\sigma_{TX}" description]
            & TTX
                \ar[d, "\mu^T_X"] \\
            MX
                \ar[rr, "\sigma_X"']
            && TX 
        \end{tikzcd}
    \]
    where $(*)$ is $M$ applied to
    \[
        \begin{tikzcd}
            X + MX
                \ar[r, "\id_X + \sigma X"]
                \ar[d, "{  [\eta^M_X,\id_{MX}]  }"']
            & X + TX
                \ar[d, "{  [\eta^T_X,\id_{TX}]  }"] \\
            MX
                \ar[r, "\sigma_X"']
            & TX
        \end{tikzcd}
    \]
    which itself commutes because on the first component it is the first axiom for $\sigma$ to be a monad morphism \eqref{eqn:monad_map_axiom1}, and on the second component it is simply $\sigma_X$ on both paths.
\end{proof}


\begin{proof}[Proof of \cref{lem:generalisation_eq_23}]
    Take $(X, \Parens{\cdot})$ that satisfies \eqref{eqn:idempotency} to \eqref{eqn:a_inside_disappear}, and a term $t \in \terms{\Sigma}{\variables}$.
    The proof goes by induction.
    When $t$ is of depth $1$, these are just equations $\eqref{eqn:a_in_front_disappear}$ and $\eqref{eqn:a_inside_disappear}$, which we already know are satisfied.
    For the induction step, take $t = \op(t_1,\ldots,t_m)$ and suppose that \eqref{eqn:generalisation2} and \eqref{eqn:generalisation3} hold for $t_1, \ldots, t_m$.
    Among the subterms, say that $p$ of them are of depth at least $1$, w.l.o.g.~the first $p$ ones $t_1,\ldots,t_p$.
    Thus, the $q \defeq m-p$ last subterms are variables $w_1,\ldots,w_q \in \set{v_1,\ldots,v_n}$.
    Then,
    \begin{align*}
        \Parens{\newa t(v_1,\ldots,v_n)}_\sigma &= \Parens{\newa ( \op(t_1,\ldots,t_p,w_1,\ldots,w_q)}_\sigma \\ 
        &= \Parens{\op(t_1,\ldots,t_p,w_1,\ldots,w_q) )}_\sigma \tag{by \eqref{eqn:a_in_front_disappear}} \\
        &= \Parens{t(v_1,\ldots,v_n)}_\sigma, 
    \end{align*}
    and
    \begin{align*}
        \Parens{&t (\newa v_1, \ldots, \newa v_n)}_\sigma \\
        &= \Parens{\op} \left(  \Parens{t_1(\newa v_1, \ldots, \newa v_n)}_\sigma, \ldots, \Parens{t_p(\newa v_1,\ldots, \newa v_n)}_\sigma, \Parens{\newa w_1}_\sigma, \ldots, \Parens{\newa w_q}_\sigma  \right) \\
        &= \Parens{\op} \left(  \Parens{t_1(v_1, \ldots, v_n)}_\sigma, \ldots, \Parens{t_p(v_1,\ldots,v_n)}_\sigma, \Parens{\newa w_1}_\sigma, \ldots, \Parens{\newa w_q}_\sigma  \right) \tag{I.H.} \\
        &= \Parens{\op (t_1,\ldots,t_p,\newa w_1, \ldots, \newa w_q)}_\sigma \\
        &= \Parens{\op (\newa t_1,\ldots, \newa t_p, \newa \newa w_1, \ldots, \newa \newa w_q)}_\sigma \tag{by \eqref{eqn:a_inside_disappear}} \\
        &= \Parens{\op (\newa t_1,\ldots, \newa t_p, \newa w_1, \ldots, \newa w_q)}_\sigma \tag{idempotency \eqref{eqn:idempotency}} \\
        &= \Parens{\op (t_1,\ldots, t_p, w_1, \ldots, w_q)}_\sigma \tag{by \eqref{eqn:a_inside_disappear}} \\
        &= \Parens{t(v_1,\ldots,v_n)}_\sigma.
    \end{align*}
\end{proof}

\begin{proof}[Proof of \cref{lem:deduction_tree_adaptation}]
    Assume $E_{} \vdash t(v_1,\ldots,v_n) = s(v_1,\ldots,v_n)$ holds, and that $T$ is a deduction tree that proves it.
    For clarity, here is a summary of the inferences rules of equational logic that we consider, where $n \in \N$, $(\op:n) \in \Sigma$, and $f$ is a substitution:
    \begin{center}
    \begin{tabular}{c c}
        {
            \AxiomC{$(s,t) \in E$}
            \RightLabel{\scriptsize Axiom$_{E}$}
            \UnaryInfC{$s=t$}
            \DisplayProof
        }
        &{
            \AxiomC{}
            \RightLabel{\scriptsize Reflexivity}
            \UnaryInfC{$t=t$}
            \DisplayProof
        } \\
        \rule{0pt}{10ex}   
        {
            \AxiomC{$s=t$}
            \RightLabel{\scriptsize Symmetry}
            \UnaryInfC{$t=s$}
            \DisplayProof
        }
        &{
            \AxiomC{$t_1=t_2$}
            \AxiomC{$t_2=t_3$}
            \RightLabel{\scriptsize Transitivity}
            \BinaryInfC{$t_1=t_3$}
            \DisplayProof
        } \\
        \rule{0pt}{10ex}   
        {
            \AxiomC{$s_1=t_1$}
            \AxiomC{$\ldots$}
            \AxiomC{$s_n=t_n$}
            \RightLabel{\scriptsize Congruence}
            \TrinaryInfC{$\op(s_1,\ldots,s_n) = \op(t_1,\ldots,t_n)$}
            \DisplayProof
        }
        &{
            \AxiomC{$s=t$}
            \RightLabel{\scriptsize Substitution}
            \UnaryInfC{$s[f]=t[f]$}
            \DisplayProof
        } 
    \end{tabular}
    \end{center}
    
    We show that a deduction tree for $E_{}\semifree \vdash t(\newa v_1,\ldots, \newa v_n) = s(\newa v_1,\ldots,\newa v_n)$ can be constructed by using $T[\newa v_i / v_i]$ as model.
    Notice first that the leaves of the deduction tree are axioms in $E_{}$ and can thus directly be adapted using equation $\eqref{eqn:a_inside_terms_disappear}$:
    given terms $t',t''$ with free variables among $v_1, \ldots,v_n$, then
    \begin{center}
        {
            \AxiomC{$(t',t'') \in E_{}$}
            \RightLabel{\scriptsize Ax.$_{E_{}}$}
            \UnaryInfC{$t' = t''$}
            \UnaryInfC{\vphantom{i}\ldots}
            \DisplayProof
        }
        $\stackrel{\eqref{eqn:a_inside_terms_disappear}}{\Longrightarrow}$
        {
            \AxiomC{$\big(t'(\newa v_1,\ldots,\newa v_n),t''(\newa v_1,\ldots,\newa v_n)\big) \in E_{}\semifree$}
            \RightLabel{\scriptsize Ax.$_{E_{}\semifree}$}
            \UnaryInfC{$t'(\newa v_1,\ldots,\newa v_n) = t''(\newa v_1,\ldots,\newa v_n)$}
            \UnaryInfC{\vphantom{i}\ldots}
            \DisplayProof
        }
    \end{center}
    Similarly, it is easy to see that each use of the reflexive, symmetry, transitive or congruence rules in our deduction tree on $E_{}$ directly translates to a version with ``$\newa$'' in front of each variable, thanks to the corresponding rule in $E_{}\semifree$.
    The complicated case to investigate is if we have an instance of the substitution rule in $T$:
    \begin{center} \label{eqn:substitution_rule_to_adapt} 
        {
            \AxiomC{$\cdots$}
            \UnaryInfC{$t'(v_1,\ldots,v_n) = t''(v_1,\ldots,v_n)$}
            \RightLabel{\scriptsize Subst.$_{E_{}}(f)$}
            \UnaryInfC{$t'(f v_1,\ldots, f v_n) = t''(f v_1,\ldots,f v_n)$}
            \UnaryInfC{$\vphantom{i}\ldots$}
            \DisplayProof
        }
    \end{center}
    for some substitution $f: \variables \to \terms{\Sigma}{\variables}$.
    This $f$ sends some variables, w.l.o.g.\ the first $p$ ones $v_1, \ldots, v_p$, to complex terms $f(v_i) = t^i(u^i_1, \ldots, u^i_{k_i})$, and sends the last $q \defeq n-p$ to variables $f(v_{p+j}) = w_j$.
    To put an ``$\newa$'' in front of every variable can be worded through the means of another substitution:
    \begin{align*}
        g(v_i) &\defeq t^i(\newa u^i_1, \ldots, \newa u^i_{k_i}), & &(1 \leq i \leq p) \\
        g(v_{p+j}) &\defeq \newa w_j, & & (1 \leq j \leq q).
    \end{align*}
    Therefore, adapting the substitution rule given above means using the premise $t'(\newa v_1,\ldots, \newa v_n) = t''(\newa v_1,\ldots,\newa v_n)$ and aiming, with the help of some deduction rules, to obtain $t'(g v_1,\ldots, g v_n) = t''(g v_1,\ldots,g v_n)$. Using two transitivity rules, we break down this into deducing three equations:
    \begin{enumerate}[(i)]
        \item $t'(gv_1,\ldots,gv_n) = t'(\newa gv_1,\ldots,\newa gv_n)$ \label{it:first}
        \item $t'(\newa gv_1,\ldots,\newa gv_n) = t''(\newa gv_1,\ldots,\newa gv_n)$ \label{it:second}
        \item $t''(\newa gv_1,\ldots,\newa gv_n) = t''(gv_1,\ldots,gv_n)$ \label{it:third}
    \end{enumerate}
    \cref{it:second} is where we find back our premise, and the rest of the tree above:
    \begin{prooftree}
        \AxiomC{$\cdots$}
        \UnaryInfC{$t'(\newa v_1,\ldots,\newa v_n) = t''(\newa v_1,\ldots,\newa v_n)$}
        \RightLabel{\scriptsize Subst.$_{E_{}\semifree}(g)$}
        \UnaryInfC{$t'(\newa gv_1,\ldots,\newa gv_n) = t''(\newa gv_1,\ldots,\newa gv_n)$}
    \end{prooftree}
    \cref{it:first}, respectively \ref{it:third}, can be easily proven by a case distinction on $t'$, respectively $t''$:
    \begin{itemize}
        \item
        If $t'$ is a variable $v_i$ among $v_1,\ldots,v_p$, then $gv_i$ is a complex term and \cref{lem:generalisation_eq_23} applies:
        \begin{prooftree}
            \AxiomC{}
            \RightLabel{\scriptsize \cref{lem:generalisation_eq_23}, \eqref{eqn:generalisation2}}
            \UnaryInfC{$gv_i = \newa gv_i$}
        \end{prooftree}
        
        \item
        If $t'$ is a variable $v_{p+j}$ among $v_{p+1},\ldots,v_n$, then idempotency proves it:
        \begin{prooftree}
            \AxiomC{}
            \RightLabel{\scriptsize Idemp. \eqref{eqn:idempotency}}
            \UnaryInfC{$\newa w_j = \newa \newa w_j$}
        \end{prooftree}
        
        \item
        If $t'$ is a complex term, then a substitution and an application of  \cref{lem:generalisation_eq_23} suffice:
        \begin{prooftree}
            \AxiomC{}
            \RightLabel{\scriptsize \cref{lem:generalisation_eq_23}, \eqref{eqn:generalisation3}}
            \UnaryInfC{$t'(v_1,\ldots,v_n) = t'(\newa v_1,\ldots,\newa v_n)$}
            \RightLabel{\scriptsize Subst.$_{E_{}\semifree}(g)$}
            \UnaryInfC{$t'(gv_1,\ldots,gv_n) = t'(\newa gv_1,\ldots,\newa gv_n)$}
        \end{prooftree}
    \end{itemize}
    A deduction tree of $E_{}\semifree \vdash t(\newa v_1,\ldots, \newa v_n) = s(\newa v_1,\ldots,\newa v_n)$ has therefore been constructed, using $T[\newa v_i / v_i]$ as model, concluding the lemma.
\end{proof}

\begin{proof}[Proof of \cref{lem:formula_a_opX}]
    We don't denote the indexes in the free monad and simply write $(T,\eta,\mu)$.
    Given any $z_1,\ldots,z_n \in TX$, we have
    \begin{align*}
        \alpha \circ \freeinterX{\op} (z_1,\ldots,z_n)
        &= \alpha \circ \freeinterX{\op} \circ (\mu_X \circ \eta_{TX})^n (z_1,\ldots,z_n)
        \tag{unit axiom \eqref{eqn:unit_axioms}} \\
        &= \alpha \circ \mu_X \circ \freeinterTX{\op} \circ \eta_{TX}^n (z_1,\ldots,z_n)
        \tag{Lem.~\ref{lem:homom_by_nat_of_mu}} \\
        &= \alpha \circ T\alpha \circ \freeinterTX{\op} \circ \eta_{TX}^n (z_1,\ldots,z_n)
        \tag{$\alpha$ assoc. \eqref{eqn:associativity_axiom2}} \\
        &= \alpha \circ \freeinterX{\op} \circ (T\alpha)^n \circ \eta_{TX}^n (z_1,\ldots,z_n)
        \tag{Lem.~\ref{lem:homom_by_nat_of_mu}} \\
        &= \alpha \circ \freeinterX{\op} \circ \eta_X^n \circ \alpha^n (z_1,\ldots,z_n)
        \tag{$\eta$ nat.}
    \end{align*}
\end{proof}

\subsection{Proofs of \cref{sec:forward_direction}}

\cref{lem:generalisation_of_def} is a direct consequence of \cref{lem:generalisation_of_def_part1} and \cref{lem:generalisation_of_def_part2} below.

\begin{lemma} \label{lem:generalisation_of_def_part1}
    For all $t \in \terms{\Sigma}{X}$ and all $\sigma : \variables \to X$:
    \[
        \alpha \circ \eta_X \circ \Parens{t}_\sigma = \alpha \circ \freeinterX{t}_{\eta_X \circ \sigma}.
    \]
\end{lemma}

\begin{proof}
    By induction on $t$:
    \begin{itemize}
        \item 
        If $t$ is some variable $v \in \variables$:
        \[
            \alpha \circ \eta_X \circ \Parens{v}_\sigma
            \stackrel{\eqref{eqn:def_[[]]_variables}}{=} \alpha \circ \eta \circ \sigma (v)
            \stackrel{\eqref{eqn:def_[[]]_variables}}{=} \alpha \circ \freeinterX{v}_{\eta_X \circ \sigma}.
        \]
        
        \item
        If $t = \op(t_1, \ldots, t_n)$:
        \begin{align*}
            \alpha \circ \eta_X \circ \Parens{t}_\sigma 
            &= \alpha \circ \eta_X \circ \Parens{\op} ( \Parens{t_1}_\sigma, \ldots, \Parens{t_n}_\sigma ) 
            \tag{def \eqref{eqn:def_[[]]_operations}} \\
            &= \alpha \circ \eta_X \circ \alpha \circ \freeinterX{\op} \circ \eta_X^n (\ldots)
            \tag{def \eqref{eqn:def_op_interpretation}} \\
            &= \alpha \circ \freeinterX{\op} \circ \eta_X^n (\ldots)
            \tag{by \eqref{eqn:a_etaX_a=a}} \\
            &= \alpha \circ \freeinterX{\op} \circ (\mu_X \circ \eta_{TX})^n \circ \eta_X^n (\ldots)
            \tag{unit axiom \eqref{eqn:unit_axioms}} \\
            &= \alpha \circ \mu_X \circ \freeinterTX{\op} \circ (\eta_{TX} \circ \eta_X)^n (\ldots)
            \tag{Lem.~\ref{lem:homom_by_nat_of_mu}} \\
            &= \alpha \circ T\alpha \circ \freeinterTX{\op} \circ (\eta_{TX} \circ \eta_X)^n (\ldots)
            \tag{$\alpha$ assoc. \eqref{eqn:associativity_axiom2}} \\
            &= \alpha \circ \freeinterX{\op} \circ (T\alpha \circ \eta_{TX} \circ \eta_X)^n (\ldots)
            \tag{Lem.~\ref{lem:homom_by_nat_of_mu}} \\
            &= \alpha \circ \freeinterX{\op} \circ (\eta_X \circ \alpha \circ \eta_X)^n ( \Parens{t_1}_\sigma, \ldots, \Parens{t_n}_\sigma )
            \tag{$\eta$ nat.} \\
            &= \alpha \circ \freeinterX{\op} \circ (\eta_X \circ \alpha)^n ( \freeinterX{t_1}_{\eta_X \circ \sigma}, \ldots, \freeinterX{t_n}_{\eta_X \circ \sigma} )
            \tag{I.H.} \\
            &= \alpha \circ \freeinterX{\op} \circ (T\alpha \circ \eta_{TX} )^n (\ldots) 
            \tag{$\eta$ nat.} \\
            &= \alpha \circ T\alpha \circ \freeinterTX{\op} \circ \eta_{TX}^n (\ldots) 
            \tag{Lem.~\ref{lem:homom_by_nat_of_mu}} \\
            &= \alpha \circ \mu_X \circ \freeinterTX{\op} \circ \eta_{TX}^n (\ldots) \tag{$\alpha$ assoc. \eqref{eqn:associativity_axiom2}} \\
            &= \alpha \circ \freeinterX{\op} \circ (\mu_X \circ \eta_{TX})^n (\ldots) 
            \tag{Lem.~\ref{lem:homom_by_nat_of_mu}} \\
            &= \alpha \circ \freeinterX{\op} ( \freeinterX{t_1}_{\eta_X \circ \sigma}, \ldots, \freeinterX{t_n}_{\eta_X \circ \sigma} )
            \tag{unit axiom \eqref{eqn:unit_axioms}} \\
            &= \alpha \circ \freeinterX{t}_{\eta_X \circ \sigma}.
            \tag{def \eqref{eqn:def_[[]]_operations}}
        \end{align*}
    \end{itemize}
\end{proof}

\begin{lemma} \label{lem:generalisation_of_def_part2}
    For all $t \in \terms{\Sigma}{X}$ of depth at least $1$ and all $\sigma : \variables \to X$:
    \[
        \alpha \circ \eta_X \circ \Parens{t}_\sigma = \Parens{t}_\sigma
    \]
\end{lemma}

\begin{proof}
    Take $t = \op(t_1,\ldots,t_n)$. 
    No induction is needed:
    \begin{align*}
        \alpha \circ \eta_X \circ \Parens{t}_\sigma
        &= \alpha \circ \eta_X \circ \Parens{\op} ( \Parens{t_1}_\sigma, \ldots, \Parens{t_n}_\sigma)
        \tag{def \eqref{eqn:def_[[]]_operations}} \\
        &= \alpha \circ \eta_X \circ \alpha \circ \freeinterX{\op} \circ \eta_X^n ( \Parens{t_1}_\sigma, \ldots, \Parens{t_n}_\sigma)
        \tag{def \eqref{eqn:def_op_interpretation}} \\
        &= \alpha \circ \freeinterX{\op} \circ \eta_X^n ( \Parens{t_1}_\sigma, \ldots, \Parens{t_n}_\sigma)
        \tag{by \eqref{eqn:a_etaX_a=a}} \\
        &= \Parens{\op} ( \Parens{t_1}_\sigma, \ldots, \Parens{t_n}_\sigma)
        \tag{def \eqref{eqn:def_op_interpretation}} \\
        &= \Parens{t}_\sigma.
        \tag{def \eqref{eqn:def_[[]]_operations}}
    \end{align*}
\end{proof}

\subsection{Proofs of \cref{sec:backward_direction}}

\begin{proof}[Proof of \cref{lem:well_definedness_of_a}]
    We need to check well-definedness of $\alpha$.
    To enhance the distinction between syntax and semantics, we use variables $v_1,\ldots,v_n$ and add the variable assignment $v_i \mapsto x_i$, i.e., we write \eqref{eqn:def_semialgebra_a} as
    \begin{equation} \label{eqn:def_semialgebra_a_old} 
        \ov{t(x_1,\ldots,x_n)} \mapsto \Parens{t(v_1,\ldots,v_n)}_{\Parens{\newa} \circ (v_i \mapsto x_i)},
    \end{equation}
    Take $c \in \freealgebra{\Sigma}{X}{E}$ and suppose that $t(x_1,\ldots,x_n)$ and $s(x_1,\ldots,x_n)$ are two representative of $c$.
    For both to be in same equivalence class means that they must be equivalent by the congruence relation generated by $E_{}$.
    Thus, $t(x_1,\ldots,x_n) = s(x_1,\ldots,x_n)$ can be deduced from $E_{}$ in equational logic.
    In other words, if we replace each $x_i$ by $v_i \in \variables$ to make it clear that we're working with syntax, we have a deduction tree of 
    \[
        E_{} \vdash t(v_1,\ldots,v_n) = s(v_1,\ldots,v_n).
    \]
    As seen in \cref{lem:deduction_tree_adaptation}, the tree can be modified to obtain a deduction tree of 
    \[
        E_{}\semifree \vdash t(\newa v_1,\ldots, \newa v_n) = s(\newa v_1,\ldots,\newa v_n)
    \]
    instead.
    Hence, by soundness, 
    \[
        E_{}\semifree \vDash t(\newa v_1,\ldots, \newa v_n) = s(\newa v_1,\ldots,\newa v_n).
    \]
    Since $(X, \Parens{\cdot})$ is a model of $E\semifree$,
    $\Parens{t (v_1,\ldots,v_n)}_{\Parens{\newa} \circ \sigma} = \Parens{s(v_1,\ldots,v_n)}_{\Parens{\newa} \circ \sigma}$
    holds for all $\sigma : \variables \to X$, and in particular for $v_i \mapsto x_i$.
    By \eqref{eqn:def_semialgebra_a_old}, $\alpha(\ov{t(x_1,\ldots,x_n)}) = \alpha(\ov{s(x_1,\ldots,x_n)})$.
    This concludes the proof that $\alpha$ is well-defined.
\end{proof}

\begin{proof}[Proof of \cref{lem:formula_a_opmx}]
    We reason as follows:
    \begin{align*}
        \alpha \circ \eta_X (x) 
        &= \alpha (\ov{x})
        \tag{def.~$\eta$} \\
        &= \Parens{x}_{\Parens{\newa}},
        \tag{def.~$\alpha$ \eqref{eqn:def_semialgebra_a}} \\
        &= \Parens{\newa} (x),
        \tag{by \eqref{eqn:def_[[]]_variables}}
    \end{align*}
    and for each $i=1,\ldots,n$ take a representative $t_i \in c_i$:
    \begin{align*} 
        \alpha \circ \freeinterX{\op} (c_1,\ldots,c_n) 
        &= \alpha \circ \freeinterX{\op} (\ov{t_1}, \ldots, \ov{t_n})  \\ 
        &= \alpha (\ov{\op(t_1,\ldots,t_n)})  \tag{def. $\freeinterX{\cdot}$} \\
        &= \Parens{ \op(t_1,\ldots,t_n) }_{\Parens{\newa}}
        \tag{def. $\alpha$ \eqref{eqn:def_semialgebra_a}} \\ 
        &= \Parens{\op} (\Parens{t_1}_{\Parens{\newa}} , \ldots , \Parens{t_n}_{\Parens{\newa}}) 
        \tag{by \eqref{eqn:def_[[]]_operations}} \\
        &= \Parens{\op} (\alpha (\ov{t_1}), \ldots, \alpha (\ov{t_m})) 
        \tag{def. $\alpha$ \eqref{eqn:def_semialgebra_a}} \\
        &= \Parens{\op} (\alpha c_1, \ldots, \alpha c_m).
    \end{align*}
\end{proof}

\subsection{Proofs of \cref{sec:ideal_monad}}

\begin{proof}[Proof of \cref{lem:analogue_semialgebra_algebra_for_ideal_monads}] 
    The proof goes similar to \cite[Theorem 3.4]{Petrisan_Sarkis_2021}.
    Any $T$-algebra $\alpha: X + T_0 X \to X$ satisfies \eqref{eqn:unit_axioms2}, i.e.,  \[\alpha \circ \eta_X = \alpha \circ \inl^{X + T_0 X} = \id_X.\]
    Hence, $\alpha$ is of the form $[id_X,a]$ for some $a : T_0 X \to X$.
    Then, the associativity axiom \eqref{eqn:associativity_axiom2}, $\alpha \circ \mu_X = \alpha \circ T \alpha$, is the following diagram:
    \begin{center}
        \begin{tikzcd}[column sep=4cm, row sep = large] 
            (X + T_0 X) + T_0 (X + T_0 X)
                \ar[r, "{\mu_X = [\Id_{TX}, \inr^{\Id+T_0} \circ m_{0,X}]}" {yshift=4pt}]
                \ar[d, "{[\Id_X,a] + T_0[\Id_X,a]}"']
            & X+T_0 X
                \ar[d, "{[\Id_X,a]}"] \\
            X + T_0 X
                \ar[r, "{[\Id_X,a]}"']
            & X
        \end{tikzcd}
    \end{center}
    Notice that in both paths, the first component is $[\id_X,a]$.
    Hence it suffices to only look at the second component:
    \[
        \begin{tikzcd}
            T_0 (X+T_0 X)
                \ar[r, "{ m_{0,X} }"]
                \ar[d, "{T_0[\Id_X,a]}"']
            & T_0 X
                \ar[r, "\inr^{\Id+T_0}"]
                \ar[dr, "a"']
                \ar[dl, phantom, "(*)" {xshift=5pt}]
            & X+T_0 X
                \ar[d, "{[\id_X, a]}"] \\
            T_0 X \ar[rr, "a"']
            & & X
        \end{tikzcd}
    \]
    The $(*)$ part of the diagram is the same as \eqref{eqn:analogue_semialgebra_algebra_for_ideal_monads}. It follows that $[id_X,a]$ is a $T$-algebra if and only if $a$ is functor $T_0$-algebra satisfying \eqref{eqn:analogue_semialgebra_algebra_for_ideal_monads}.
    
    We therefore have mappings $[\id_X,a] \mapsto a$ and $a \mapsto [\id_X,a]$ on objects.
    To have an isomorphism of categories, it suffices to show that for homomorphisms:
    \[
        \begin{tikzcd}
            X + T_0 X
                \ar[r, "f + T_0 f"]
                \ar[d, "{[\id_X,a]}"']
            & Y + T_0 Y
                \ar[d, "{[\id_X,b]}"] \\
            X
                \ar[r, "f"']
            & Y
        \end{tikzcd}
        \qquad \Leftrightarrow \qquad
        \begin{tikzcd}
            T_0 X
                \ar[r, "T_0 f"]
                \ar[d, "{a}"']
            & T_0 Y
                \ar[d, "{b}"] \\
            X
                \ar[r, "f"']
            & Y
        \end{tikzcd}
    \]
    Both directions clearly hold.
    We have an isomorphism of categories as desired, concluding the proof.
\end{proof}

\begin{proof}[Proof of \cref{lem:semifree_construction_not_pointed_endofunctor}]
    We prove that $(-)\semifree$ is not pointed.
    Suppose for contradiction that it is, i.e., there exists a natural transformation $\tau : \id_{\catmon{C}} \Rightarrow (-)\semifree$.
    Consider the final $\catset$-monad $\singleset$, which is defined as $\singleset(X)
    \defeq 1 = \singl$ for sets $X$ and $\singleset(f) \defeq \id_1 : 1 \to 1$ for functions $f$.
    We thus have a monad morphism ${\tau_{\singleset} : \singleset \Rightarrow \singleset\semifree }$. 
    On the empty set, $\tau_{\singleset, \emptyset} : 1 \to \emptyset + 1$ is necessarily $\inr^{\emptyset,1}$.
    The naturality of $\tau_{\singleset}$ implies that $\tau_{\singleset,X}$ is $\inr^{X,1}$ for all sets $X$:
    
    \noindent\begin{minipage}{.5\linewidth}
        \begin{center}
            \begin{tikzcd}
                \singleset(\emptyset) = 1
                    \ar[r, "\inr^{\emptyset,1}"]
                    \ar[d, "\singleset(\emptyset) = \id_1"']
                & \emptyset + 1
                    \ar[d, "\emptyset + \id_1 = \singleset\semifree(\emptyset)"] \\
                \singleset(X) = 1
                    \ar[r, "\tau_{\singleset, X}"']
                & X+1
            \end{tikzcd}
        \end{center}
    \end{minipage}
    \hfill
    \begin{minipage}{.5\linewidth}
        \begin{align*}
            \tau_{\singleset,X}(\ast) 
            &= \tau_{\singleset,X} (\id_1(\ast)) \\
            &= (\emptyset + \id_1)(\inr^{\emptyset,1}(\ast)) \\
            &= \inr^{X,1}(\ast) \\
        \end{align*}
    \end{minipage}
    But then, the first monad morphism axiom \eqref{eqn:monad_map_axiom1} fails, which is a contradiction:\\
    \hspace*{10em}
        \begin{tikzcd}[row sep=tiny]
            & \singleset(X) = 1
                \ar[dd, "\tau_{\singleset,X} = \inr^{X+1}"] \\
            X
                \ar[ur, "\eta_X = 1"]
                \ar[dr, "\eta\semifree_X = \inl^{X+1}"']
            & \\
            & \singleset \semifree (X) = X+1
        \end{tikzcd}
\end{proof}    
    
\begin{proof}[Proof of \cref{lem:copointed_endofunctor}]    
    We prove that $(-)\semifree : \catmon{C} \to \catmon{C}$ is a comonad when endowed with the following counit $\epsilon$ and comultiplication $\delta$:
    For $(M,\eta,\mu) \in \catmon{C}$ and $X \in \cat{C}$, let
    \begin{align*}
        \epsilon_{M,X} : X+MX &\xrightarrow{[\eta_X,\id_{MX}]} MX,  \\
        \delta_{M,X} : X + MX &\xrightarrow{\id_X + \inr^{X+MX}} X + (X + MX).
    \end{align*}
    We must check the following points:
    \begin{enumerate}
        \item
        Well-definedness of $\epsilon$: Given $M \in \catmon{C}$, we check that $\epsilon_M : M\semifree \Rightarrow M$ is a monad morphism, i.e.,
        
        \begin{enumerate}
            \item 
            Naturality of $\epsilon_M$: It follows from the naturality of $\eta$.
            Given $f: X \to Y$,
            \[
                \begin{tikzcd}
                    X
                        \ar[r, "\eta_X"]
                        \ar[d, "f"']
                        \ar[dr, phantom, "{\scriptstyle (\eta \text{ nat.})}"]
                    & MX
                        \ar[d, "Mf"] \\
                    Y
                        \ar[r, "\eta_Y"']
                    & MY
                \end{tikzcd}
                \qquad \Rightarrow \qquad
                \begin{tikzcd}[column sep=6em]
                    X+MX
                        \ar[r, "\epsilon_{M,X} = {[\eta_X,\id_{MX}]}"]
                        \ar[d, "f+Mf"']
                    & MX
                        \ar[d, "Mf"] \\
                    Y+MY
                        \ar[r, "\epsilon_{M,Y} = {[\eta_Y,\id_{MY}]}"']
                    & MY
                \end{tikzcd}
            \]
            
            \item
            The monad morphism axioms are satisfied:
            The first one is simply
            \[
                \begin{tikzcd}[row sep=small]
                    & X+MX
                        \ar[dd, "\epsilon_{M,X} = {[\eta_X,\id_{MX}]}"] \\
                    X
                        \ar[ur, "\eta\semifree_X =  \inl^{X+MX}"]
                        \ar[dr, "\eta_X"']
                    & \\
                    & MX
                \end{tikzcd}
            \]
            It clearly commutes.
            The second one is 
            \[
                \begin{tikzcd}[scale cd=.85]
                    (X+MX) + M(X+MX)
                        \ar[r, "{    [\eta_X,\id_X] + M[\eta_X,\id_X]    }"{yshift=4pt}]
                        \ar[d, "{  \mu\semifree_X = [\id_{X+MX}   ,  \inr^{X+MX} \circ \mu_X \circ M[\eta_X,\id_{MX}]   ]   }" description]
                    & MX + MMX
                        \ar[r, "{   [\eta_{MX},\id_{MMX}]   }"{yshift=4pt}]
                    & MMX
                        \ar[d, "\mu_X \phantom{\mu\semifree_X = [\id_{X+MX}   ,  \inr^{X+MX} }"] \\
                    X + MX
                        \ar[rr, "{   [\eta_X,\id_X]   }"']
                    &
                    & MX
                \end{tikzcd}
            \]
            To see that it commutes, let us look at each component.
            They are respectively:
            \begin{center}
                \begin{tikzcd}[column sep= small, scale cd=.75]
                    X
                        \ar[r, "\eta_X"]
                        \ar[d, equal]
                    & MX
                        \ar[r, "\eta_{MX}"]
                        \ar[dr, equal]
                    & MMX
                        \ar[d, "\mu_X"] \\
                    X
                        \ar[rr, "\eta_X"']
                        \ar[urr, phantom, "\small \eqref{eqn:unit_axioms}" {yshift=-4pt}, pos=.99]
                    && MX
                \end{tikzcd}
                \quad
                \begin{tikzcd}[scale cd=.75]
                    MX
                        \ar[r, "\eta_{MX}"]
                        \ar[d, equal]
                        \ar[dr, phantom, "\small\eqref{eqn:unit_axioms}"]
                    & MMX
                        \ar[d, "\mu_X"] \\
                    MX
                        \ar[r, equal]
                    & MX
                \end{tikzcd}
                \quad
                \begin{tikzcd}[scale cd=.75] 
                    M(X+MX)
                        \ar[r, "{   M[\eta_X,\id_{MX}]   }"]
                        \ar[d, "{   \mu_X  \circ  M[\eta_X,\id_{MX}]   }" description]
                    & MMX
                        \ar[d, "\mu_X"] \\
                    MX
                        \ar[r, equal]
                    & MX
                \end{tikzcd}
            \end{center}
        \end{enumerate}
        
        \item
        Naturality of $\epsilon$: Given a monad morphism $\sigma : M \Rightarrow T$, we instantiate with an object $X$ the diagram that we must check:
        \begin{center}
            \begin{tikzcd}[column sep = 6em]
                M\semifree X = X + MX
                    \ar[r, "\epsilon_{M,X} = {  [\eta^M_X , \id_{MX}]  }"]
                    \ar[d, "\sigma\semifree_X = \id_X + \sigma_X"']
                & MX
                    \ar[d, "\sigma_X"]
                \\
                T\semifree X = X + TX
                    \ar[r, "\epsilon_{T,X} = {  [\eta^T_X , \id_{TX}]  }"']
                & TX
            \end{tikzcd}
        \end{center}
        It commutes since the first component is the first monad morphism axiom \eqref{eqn:monad_map_axiom1} for $\sigma$, and the second component is $\sigma_X$ on both paths.
        
        \item
        Well-definedness of $\delta$: Given $M \in \catmon{C}$, we check that $\delta_M : M\semifree \Rightarrow M\nsemifree{2}$ is a monad morphism, i.e.,
        
        \begin{enumerate}
            \item
            Naturality of $\delta_M$: Given $f : X \to Y$,
            \begin{center}
                \begin{tikzcd}[column sep = 8em]
                    M\semifree X = X + MX
                        \ar[r, "\delta_{M,X} = \id_X + \inr^{X+MX}"]
                        \ar[d, "f + Mf"']
                    & X + (X + MX) = M \nsemifree{2} X
                        \ar[d, "f + (f + Mf)"]
                    \\
                    M\semifree Y = Y + MY
                        \ar[r, "\delta_{M,Y} = \id_Y + \inr^{Y+MY}"']
                    & Y + (Y + MY) = M \nsemifree{2} Y
                \end{tikzcd}
            \end{center}
            It commutes, since the first component is $f$ on both paths, and the second component is a known coproduct formula.
            
            \item
            The monad morphism axioms are satisfied:
            The first one is simply
            \[
                \begin{tikzcd}[row sep=small]
                    & M\semifree X = X + MX
                        \ar[dd, "\delta_{M,X} = \id_X + \inr^{X + MX}"] \\
                    X
                        \ar[ur, "\eta\semifree_X =  \inl^{X + MX}"]
                        \ar[dr, "\eta\nsemifree{2}_X = \inl^{X + M\semifree X}"']
                    & \\
                    & M \nsemifree{2} X = X + (X + MX)
                \end{tikzcd}
            \]
            It clearly commutes.
            The second one is 
            \[
                \begin{tikzcd}[scale cd=.85]
                    M \semifree M \semifree X
                        \ar[r, "{  M\semifree \delta_{M,X}  }"] 
                        \ar[d, "{  \mu\semifree_X  }"' ]
                    & M \semifree M \nsemifree{2} X
                        \ar[r, "{  \delta_{M \nsemifree{2} X}  }"] 
                    & M \nsemifree{2} M \nsemifree{2} X
                        \ar[d, "\mu \nsemifree{2} X"] \\
                    M \semifree X
                        \ar[rr, "{  \delta_{M,X}  }"']
                    &
                    & M \nsemifree{2} X
                \end{tikzcd}
            \]
            The diagram with every object and every morphism fully detailed does not fit.
            We therefore look directly at each component.
            The domain is $(X + MX) + M(X + MX)$.
            On the first component  $X + MX$, it is the identity on both path.
            On the second component $M(X + MX)$, we have the following, wich commute:
            \begin{center}
                \begin{tikzcd}[scale cd=.75, column sep=6em] 
                    M(X+MX)
                        \ar[r, "M(\id_X + \inr^{X + MX})"]
                        \ar[dd, "{   M[\eta_X,\id_{MX}]   }"']
                        \ar[dr, equal]
                    & M(X + (X + MX)) 
                        \ar[d, "{  M[\inl^{X + MX} , \id_{X + MX}]  }"]
                    \\
                    {}
                    & M(X + MX)
                        \ar[d, "{  M[\eta_X, \id_{MX}]  }"]
                    \\
                    MMX
                        \ar[d, "\mu_X"']
                    & MMX
                        \ar[d, "\mu_X"]
                    \\
                    MX
                        \ar[r, equal]
                    & MX
                \end{tikzcd}
            \end{center}
        \end{enumerate}
        
        \item
        Naturality of $\delta$: Given a monad morphism $\sigma : M \Rightarrow T$, we instantiate with a set $X$ the diagram that we must check:
        \begin{center}
            \begin{tikzcd}[column sep = 8em]
                M\semifree X = X + MX
                    \ar[r, "\id_X + \inr^{X+MX}"]
                    \ar[d, "\id_X + \sigma_X"']
                & X + (X + MX) = M\nsemifree{2} X
                    \ar[d, "\id_X + (\id_X + \sigma_X)"]
                \\
                T\semifree X = X + TX
                    \ar[r, "\id_X + \inr^{X+TX}"']
                & X + (X + TX) = T\nsemifree{2} X
            \end{tikzcd}
        \end{center}     
        It commutes, since the first component is $f$ on both paths, and the second component is a known coproduct formula.
        
        \item
        The comonad axioms: We must check the following
        
        \noindent\begin{tabularx}{\textwidth}{@{}XX@{}}
            \begin{equation*}
                \begin{tikzcd}[ampersand replacement=\&]
                    \& (-)\semifree
                        \ar[dl, equal]
                        \ar[dr, equal]
                        \ar[d, "\delta"]
                    \& \\
                    (-)\semifree
                    \& (-)\nsemifree{2}
                        \ar[l, "(-)\semifree \epsilon"]
                        \ar[r, "\epsilon (-)\semifree"']
                    \& (-)\semifree
                \end{tikzcd}
            \end{equation*} &
            \begin{equation*}
                \begin{tikzcd}[ampersand replacement=\&]
                    (-)\semifree
                        \ar[r, "\delta"]
                        \ar[d, "\delta"']
                    \& (-)\nsemifree{2}
                        \ar[d, "\delta (-)\semifree"]
                    \\
                    (-)\nsemifree{2}
                        \ar[r, "(-)\semifree \delta"']
                    \& (-)\nsemifree{3}
                \end{tikzcd}
            \end{equation*}
        \end{tabularx}
        We instantiate with a monad $M$ and a set $X$ and indeed,
        \begin{align*}
            \epsilon_{M,X}\semifree \circ \delta_{M,X} 
            &= \left(\id_X + [\eta_X, \id_{MX}]\right) \circ \left(  \id_X + \inr^{X + MX}  \right) \\
            &= \id_X + \id_{MX} \\
            &= \id_{M\semifree X},
            \intertext{and}
            \epsilon_{M\semifree,X} \circ \delta_{M,X}
            &= \left(  \id_X + \inr^{X + MX}  \right) \circ [\eta\semifree_X = \inl^{X+MX}, \id_{M\semifree X}] \\
            &= [\inl^{X+MX}, \inr^{X+MX}] \\
            &= \id_{M\semifree X}.
        \end{align*}
        Also,
        \begin{align*}
            \delta_{M,X} \semifree \circ \delta_{M,X} 
            &= \left( \id_X + (\id_X + \inr^{X+MX}) \right) \circ \left(  \id_X + \inr^{X+MX}  \right) \\
            &= \id_X + (\inr^{X+(X+MX)} \circ \inr^{X+MX}) \\
            &= \left(  \id_X + \inr^{X+(X+MX)}  \right) \circ \left(  \id_X + \inr^{X+MX}  \right) \\
            &= \delta_{M\semifree, X} \circ \delta_{M,X}.
        \end{align*}
    \end{enumerate}
    Everything has been checked, and the proof is thus complete.
\end{proof}

\end{document}